\documentclass[a4paper,UKenglish,cleveref]{lipics-v2019-2}

\usepackage[utf8]{inputenc}
\usepackage{graphicx}
\usepackage{thm-restate}
\usepackage{amsmath,amssymb,xcolor,graphicx}
\usepackage{tikz}
\usetikzlibrary{fit,automata,arrows}
\usetikzlibrary{patterns,hobby}
\usepackage{caption}
\usepackage[labelformat=simple]{subcaption}
\usepackage{pgfplots}
\usepackage{enumerate}
\usepackage{mathtools}
\usepackage{float}
\usepackage{cleveref}

\restylefloat{table}

\newtheorem{open}{Open Problem}

\theoremstyle{definition}
\newtheorem{observation}[theorem]{Observation}

\newtheoremstyle{parens}
  {}
  {}
  {\itshape}
  {}
  {}
  {}
  {.5em}
  {(\thmnumber{#2})}

\theoremstyle{parens}
\newtheorem{nitem}[equation]{}

\newcommand{\NP}{{\sf NP}}

\newcommand{\FPT}{{\sf FPT}}
\newcommand{\W}{{\sf W}}

\renewcommand\cref[1]{\Cref{#1}}

\DeclareMathOperator{\thin}{thin}
\DeclareMathOperator{\pthin}{pthin}

\newcommand{\AW}{{\sf AW}[*]}
\newcommand{\adj}[2]{\mathrm{Adj}_{#1}(#2)}
\newcommand{\gr}{\mathrm{gr}}
\newcommand{\start}{{\sf start}}

\usetikzlibrary{shapes,shapes.geometric,arrows,fit,calc,positioning,automata,arrows}
\usetikzlibrary{patterns,hobby}
\usepackage{pgfplots}
\pgfplotsset{compat=1.15}

\tikzstyle{line}=[draw]

\tikzset{ n1/.style={circle,scale=1.0},
n2/.style={circle,fill=black,scale=0.5},
n3/.style={circle,draw,fill=black,draw=black,text=white,scale=0.9},
e1/.style={line width=0.1mm}, e3/.style={draw=black,line
width=0.7mm}, inter/.style={line width=0.85mm},
e2/.style={draw=black,line width=0.5mm}, c1/.style={line
width=0.3mm}, c2/.style={line width=0.2mm} }

\def\centerarc[#1](#2)(#3:#4:#5)%
    { \draw[#1] ($(#2)+({#5*cos(#3)},{#5*sin(#3)})$) arc (#3:#4:#5); }

\newcommand{\vertex}[4][black]{
    \draw[#1, fill=#1, inner sep=0pt] (#2, #3) circle (0.12) node(#4){};
}

\newcommand{\vertexLabel}[3][above]{
    \path (#2) node[#1]{#3};
}

\nolinenumbers
\title{Non-crossing $H$-graphs: a generalization of proper interval graphs admitting FPT algorithms}

\titlerunning{Non-crossing $H$-graphs}

\author{Flavia Bonomo-Braberman}{Universidad de Buenos Aires. Facultad de Ciencias Exactas y Naturales. Departamento de Computaci\'on. Buenos Aires, Argentina. / CONICET-Universidad de Buenos Aires. Instituto de
Investigaci\'on en Ciencias de la Computaci\'on (ICC). Buenos
Aires, Argentina}{fbonomo@dc.uba.ar}{https://orcid.org/0000-0002-9872-7528}{}
\author{Nick Brettell}{School of Mathematics and Statistics, Victoria University of Wellington, New Zealand}{nick.brettell@vuw.ac.nz}{https://orcid.org/0000-0002-1136-418X}{}
\author{Noleen K\"{o}hler}{School of Computer Science, University of Leeds, UK}{n.koehler@leeds.ac.uk}{https://orcid.org/0000-0002-1023-6530}{}
\author{Andrea Munaro}{Department of Mathematical, Physical and Computer Sciences, University of Parma, Italy}{andrea.munaro@unipr.it}{https://orcid.org/0000-0003-1509-8832}{}
\author{Dani\"el Paulusma}{Department of Computer Science, Durham University, UK}{daniel.paulusma@durham.ac.uk}{https://orcid.org/0000-0001-5945-9287}{}

\authorrunning{F. Bonomo-Braberman, N. Brettell, N. K\"{o}hler, A. Munaro, D. Paulusma}


\acknowledgements{Bonomo was supported by  UBACyT (20020220300079BA and 20020190100126BA), CONICET (PIP 11220200100084CO), and ANPCyT (PICT-2021-I-A-00755); Brettell by the New Zealand Marsden Fund; and Paulusma by the Leverhulme Trust (RPG-2016-258). Munaro is a member of the Gruppo Nazionale Calcolo Scientifico-Istituto Nazionale di Alta Matematica (GNCS-INdAM).
An extended abstract containing part of the results of this paper will appear in the proceedings of WG 2025~\cite{BBMP25}.}

\hideLIPIcs

\begin{document}

\maketitle
\begin{abstract}
We prove new parameterized complexity results for the FO Model Checking problem on a well-known generalization of interval and circular-arc graphs: the class of $H$-graphs, for any fixed multigraph~$H$. In particular, we research how the parameterized complexity differs between two subclasses of $H$-graphs: proper $H$-graphs and non-crossing $H$-graphs, each generalizing proper interval graphs and proper circular-arc graphs. We first 
generalize a known result of Bonnet et al.~(IPEC 2022) from interval graphs to $H$-graphs, for any (simple) forest~$H$, by showing that for such $H$, 
the class of $H$-graphs is delineated. This implies that for every hereditary subclass~${\cal D}$ of $H$-graphs, FO Model Checking is in \FPT\ if ${\cal D}$ has bounded twin-width and $\AW$-hard otherwise. As proper claw-graphs  have unbounded twin-width, this means that FO Model Checking is $\AW$-hard for proper $H$-graphs for certain forests~$H$ like the claw. In contrast, we show that
even for every multigraph~$H$, non-crossing $H$-graphs have bounded proper mixed-thinness and hence bounded twin-width, and thus FO Model Checking is in $\FPT$ on non-crossing $H$-graphs when parameterized by $\Vert H \Vert+\ell$, where $\Vert H \Vert$ is the size of $H$ and $\ell$ is the size of a formula. It is known that a special case of FO Model Checking,  {\sc Independent Set}, is
$\mathsf{W}[1]$-hard on $H$-graphs  when parameterized by $\Vert H \Vert +k$, where $k$ is the size of a solution. We strengthen this $\mathsf{W}[1]$-hardness result to proper $H$-graphs. Hence, we solve, in two different ways,  an open problem of Chaplick (Discrete Math.~2023), who asked about problems that can be solved faster for non-crossing $H$-graphs than for proper $H$-graphs.
\end{abstract}

\section{Introduction}\label{s-intro}

The FO Model Checking problem is to decide whether a given graph $G$ satisfies a given first-order formula~$\phi$. As such, the FO Model Checking problem captures a large class of graph problems, including {\sc Clique}, {\sc Dominating Set} and {\sc Independent Set} and so on. We prove new parameterized complexity results for FO Model Checking by restricting the input to certain classes of intersection graphs and proving boundedness of twin-width. Showing the latter will enable us to apply a well-known result of Bonnet et al.~\cite{BKTW22}, which states that FO Model Checking is in \FPT\ for graphs of bounded twin-width, 
provided a contraction sequence witnessing small twin-width is given.
We will also show several new hardness results. Before explaining this in more detail, we first introduce our framework.

\medskip
\noindent
{\bf Our Framework.}
We consider the framework of intersection graphs, because this framework captures many well-known graph classes and plays a central role in algorithmic graph theory. An {\it intersection graph}~$G_{\cal S}$ of a family~${\cal S}$ of subsets of a set~$X$ has a unique vertex~$s_i$ for each $S_i\in {\cal S}$ and there is an edge between two distinct vertices $s_i$ and~$s_j$ if and only if $S_i\cap S_j\neq \varnothing$. Two very popular classes of intersection graphs are those of {\it interval graphs} and {\it circular-arc graphs}. We obtain these classes by letting ${\cal S}$ be a family of intervals of the real line or arcs of the circle, respectively. One of the reasons that both graph classes are so widely known is that they have various ``good'' algorithmic properties, often leading to fast algorithms.

It is natural to try to extend a graph class~${\cal G}$ into a larger graph class that preserves the desirable algorithmic properties of ${\cal G}$. Applying this approach to interval graphs and circular-arc graphs has led to the notion of $H$-graphs, introduced in 1992 by B\'{\i}r\'o, Hujter and Tuza~\cite{B-H-T-mp}. In order to define $H$-graphs, we first need to introduce some terminology.
The {\it subdivision} of an edge $uv$ in a graph replaces~$uv$ by a new vertex~$w$ and edges~$uw$ and $wv$. A graph~$H'$ is a {\it subdivision} of a graph~$H$ if we can modify $H$ into $H'$ by a sequence of edge subdivisions. A vertex subset of a graph is {\it connected} if it induces a connected subgraph (the empty set is also connected).
A {\it multigraph} may have parallel edges (edges between the same pair of vertices) and self-loops (edges $uu$). A {\it simple} graph is a graph with no parallel edges and no self-loops.

\begin{definition}
For a multigraph~$H$, a (simple) graph $G$ is an \emph{$H$-graph} if $G$ is the intersection graph of a family~${\cal S}$ of connected vertex subsets of a subdivision~$H'$ of $H$;
we say that ${\cal S}$ is an \emph{$H$-representation} of~$G$.
\end{definition}

\noindent
We refer to Figure~\ref{fig:proper-non-crossing} for an example of an $H$-graph~$G$. We also
note that the $P_2$-graphs are exactly the interval graphs, and that the $C_2$-graphs are exactly the circular-arc graphs;
here, $P_r$ is the $r$-vertex path and $C_s$ is the $s$-vertex cycle, where $C_2$ is the $2$-vertex cycle consisting of two vertices with two parallel edges. As such, the size $\Vert H \Vert=|E(H)|$ (counting multiplicities) of $H$ is a natural parameter for ``measuring''
how far an $H$-graph is from being interval or circular-arc. Hence, $H$-graphs form a {\it parameterized} generalization of  interval graphs and circular-arc graphs. This led Chaplick et al.~\cite{ChaplickTVZ21} to initiate an algorithmic study of $H$-graphs. They showed that fundamental problems, such as \textsc{Independent Set} and \textsc{Dominating Set}, are all in $\mathsf{XP}$ on $H$-graphs, when parameterized by $\Vert H \Vert$, and also that {\sc Clique} is para-$\mathsf{NP}$-hard on $H$-graphs when parameterized by $\Vert H\Vert$ (see~\cite{CFK,FG} for background on parameterized complexity).
Fomin, Golovach and Raymond~\cite{Algo-H-graphs} complemented these results by showing that \textsc{Independent Set} and
\textsc{Dominating Set} on $H$-graphs are $\mathsf{W}[1]$-hard when parameterized by
$\Vert H \Vert +k$, where $k$ denotes the size of the independent (resp. dominating) set. We refer to~\cite{Cagir23-recog-H-graphs,ChaplickTVZ21,Klav-T-NPc,Tuc-test} for \NP-completeness and polynomial-time results for recognizing $H$-graphs, depending on the structure of $H$.

We obtain the classes of {\it proper} interval graphs and {\it proper} circular-arc graphs by requiring the corresponding family ${\cal S}$ of intervals of the real line or arcs of the circle, respectively, to be {\it proper}; here, a family ${\cal S}$ is \emph{proper} if no set in ${\cal S}$ properly contains another.
Proper interval graphs and proper circular-arc graphs have certain structural properties that interval graphs and circular-arc graphs do not have, and that can be exploited algorithmically.
In particular, FO Model Checking is well known to be in $\mathsf{FPT}$ on proper interval graphs~\cite{FO-mc-interval} and even on proper circular-arc graphs~\cite{HPR19} when parameterized by the length of the formula. As explained in~\cite{HPR19}, such a result does not hold for interval graphs. This is due to the fact that
 {\sc Induced Subgraph Isomorphism} (does a given graph $G$ contain a given graph $H$ as an induced subgraph?) is \W[1]-hard on interval graphs when parameterized by $|V(H)|$~\cite{MS13}.
 
The underlying reason a problem is ``efficiently solvable'' on some graph class ${\cal G}$ might be the fact that ${\cal G}$ has {\it bounded} width for some width parameter $p$, that is, there exists a constant $c$ such that  $p(G)\leq c$ for every graph $G\in {\cal G}$.
Bonnet et al.~\cite{BKTW22} proved that FO Model Checking is in $\mathsf{FPT}$ when parameterized by twin-width plus the length of the formula.
As it turns out, proper circular-arc graphs, and thus proper interval graphs, have bounded twin-width~\cite{Jedelsky2021thesis}, whereas  interval graphs, and thus circular-arc graphs and more generally $H$-graphs, have unbounded twin-width~\cite{BGKTW22}.
This raises the question whether it is possible to generalize proper interval graphs and proper circular-arc graphs to  subclasses of $H$-graphs whose twin-width is at most $f(\Vert H \Vert)$ for some function $f$ that only depends on $\Vert H \Vert$. Due to~\cite{BKTW22}, the latter would imply that FO Model Checking is in $\mathsf{FPT}$ when parameterized by $\Vert H \Vert$ plus the length of the formula.

Recently, two new rich subclasses of $H$-graphs were introduced. First, Chaplick et al.~\cite{CFGKZ21} defined proper $H$-graphs (see also Figure~\ref{fig:proper-non-crossing}).

\begin{definition}
For a multigraph~$H$, a (simple) graph $G$ is a \emph{proper} $H$-graph if $G$ is the intersection graph of a proper family~${\cal S}$ of connected vertex subsets of a subdivision~$H'$ of $H$;
we say that ${\cal S}$ is a \emph{proper} $H$-representation of~$G$.
\end{definition}

 \begin{figure}[t]
     \centering
     \begin{subfigure}[b]{.2\textwidth}
         \centering
        \begin{tikzpicture}[scale=.7]
        \draw [gray] (-2,0) -- (1.8,0);
        \draw [gray] (0,0) -- (0,-1.6);
        \draw  (-1.2,0.2) -- (0.8,0.2);
        \draw  (-0.2,0.2) -- (-0.2,-0.8);
            \draw  (-1.8,0.1) -- (-0.8,0.1);
            \draw  (0.4,0.1) -- (1.4,0.1);
        \draw  (-0.1,-0.4) -- (-0.1,-1.4);
        \end{tikzpicture}
         \caption{\footnotesize proper and non-crossing \\ \phantom{new line}}
        \label{fig:p-nc}
     \end{subfigure}
     \hfill
     \begin{subfigure}[b]{.2\textwidth}
         \centering
          \begin{tikzpicture}[scale=.7]
            \draw [gray] (-2,0) -- (1.8,0);
        \draw [gray] (0,0) -- (0,-1.6);
            \draw  (-1.2,0.2) -- (1.4,0.2);
        \draw  (-0.2,0.2) -- (-0.2,-0.8);
            \draw  (-1.8,0.1) -- (-0.8,0.1);
            \draw  (0.4,0.1) -- (1.4,0.1);
        \draw  (-0.1,-0.4) -- (-0.1,-1.4);
       \end{tikzpicture}
         \caption{\footnotesize non-crossing but not proper \\ \phantom{new line}}
        \label{fig:np-nc}
     \end{subfigure}
     \hfill
     \begin{subfigure}[b]{.2\textwidth}
         \centering
        \begin{tikzpicture}[scale=.7]
        \draw [gray] (-2,0) -- (1.8,0);
        \draw [gray] (0,0) -- (0,-1.6);
        \draw  (-1.2,0.2) -- (1,0.2);
        \draw  (-0.1,0.1) -- (-0.1,-0.8);
            \draw  (-1.8,0.1) -- (-0.8,0.1);
            \draw  (0.6,0.1) -- (1.6,0.1);
            \draw  (-0.6,0.1) -- (0.4,0.1);
        \end{tikzpicture}
         \caption{\footnotesize proper but not non-crossing \\ \phantom{new line}}
        \label{fig:p-nnc}
     \end{subfigure}
     \hfill
     \begin{subfigure}[b]{.2\textwidth}
         \centering
        \begin{tikzpicture}[scale=.7]
      \draw [gray] (-2,0) -- (1.8,0);
      \draw [gray] (0,0) -- (0,-1.6);
      \draw  (-1.2,0.2) -- (1,0.2);
            \draw  (-1.8,0.1) -- (-0.8,0.1);
            \draw  (0.6,0.1) -- (1.6,0.1);
            \draw  (-0.6,0.1) -- (0.4,0.1);
        \end{tikzpicture}
         \caption{\footnotesize neither proper nor non-crossing \\ \phantom{new line}}
        \label{fig:np-nnc}
     \end{subfigure}
            \vspace*{-0.3cm}
        \caption{Four different representations of the claw ($K_{1,3}$) as a claw-graph.}
        \label{fig:proper-non-crossing}
        \vspace*{-0.3cm}
\end{figure}
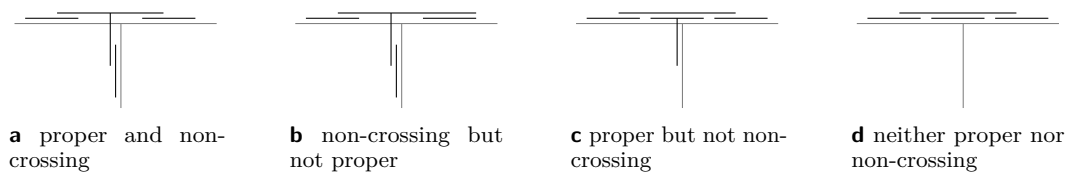

\noindent
Note that the proper $P_2$-graphs are exactly the proper interval graphs, and that the proper $C_2$-graphs are exactly the proper circular-arc graphs.
We refer to~\cite{Chap-proper-tree,D-H-H-circ-arc,Klav-T-NPc} for \NP-completeness and polynomial-time results for recognizing proper $H$-graphs, depending on $H$.
Chaplick et al.~\cite{CFGKZ21} showed the existence of polynomial kernels for Hamiltonian-type problems on proper $H$-graphs.
However, as a possible drawback, Jedelsk\'y~\cite{Jedelsky2021thesis} proved that already proper $K_{1,3}$-graphs have unbounded twin-width; here, $K_{1,3}$ denotes the {\it claw}, i.e., the $4$-vertex star.

Second, Chaplick~\cite{Chap-nc-paths} defined non-crossing $H$-graphs, whose definition requires an extra notion. Namely, a family~${\cal S}$ of connected vertex subsets of some graph is
\emph{non-crossing} if, for every pair $A$, $B$ in ${\cal S}$, both $A\setminus B$ and $B\setminus A$ are connected.
We again refer to Figure~\ref{fig:proper-non-crossing} for an example of non-crossing $H$-graphs (see~\cite{Kr96} for a different type of non-crossing intersection graphs that includes disk graphs).

\begin{definition}
For a multigraph~$H$, a (simple) graph $G$ is a \emph{non-crossing} $H$-graph if $G$ is the intersection graph of a non-crossing family~${\cal S}$ of connected vertex subsets of a subdivision~$H'$ of $H$;
we say that ${\cal S}$ is a \emph{non-crossing} $H$-representation of~$G$.
\end{definition}

\noindent
Just like proper $P_2$-graphs, non-crossing $P_2$-graphs are exactly proper interval graphs~\cite{Chap-nc-paths}.
However, unlike proper $C_2$-graphs, non-crossing 
$C_2$-graphs and proper circular-arc graphs do not coincide: 
the claw is an example of a non-crossing $C_2$-graph that is not proper circular-arc. Hence, 
proper $C_2$-graphs form a proper subclass of the class of non-crossing $C_2$-graphs.
On the other hand, we will show in Section~\ref{s-preliminaries} that there exists a proper $K_{1,3}$-graph, namely the $4$-fan, that is not a non-crossing $T$-graph for any tree~$T$.
So, proper $H$-graphs and non-crossing $H$-graphs form incomparable graph classes in general.

Chaplick~\cite{Chap-nc-paths} gave structural properties of non-crossing $H$-graphs and variants thereof to obtain linear-time algorithms not only for recognizing certain non-crossing $H$-graphs,
but also, for example, for deciding if  a chordal graph is claw-free. Given their structure, Chaplick~\cite{Chap-nc-paths}  asked the following:

 \medskip
 \noindent
 {\it Can some problems be solved faster for non-crossing $H$-graphs than for proper $H$-graphs?}

\medskip
\noindent
We aim to answer this question with respect to the FO Model Checking problem, for which we recall that boundedness of twin-width suffices~\cite{BGKTW22}, and we will therefore also focus on:

\medskip
\noindent
{\it Do non-crossing $H$-graphs have twin-width at most $f(\Vert H \Vert)$ for some function $f$ that only depends on $\Vert H \Vert$?}

\medskip
\noindent
{\bf Our Results.}
We provide positive answers to both questions in various ways.
We first let $H$ be a (simple) forest, and show that the parameterized complexity of FO Model Checking on proper $H$-graphs and non-crossing $H$-graphs may already be different for such~$H$, e.g. if $H=K_{1,3}$. Afterwards, we consider general multigraphs~$H$ and show again a difference in parameterized complexity for proper $H$-graphs and non-crossing $H$-graphs, but in a different way.

When $H$ is a forest, we will use
the concept of delineation (by twin-width), which was introduced by Bonnet et al.~\cite{BonnetC0K0T22} to classify exactly those hereditary subclasses of some graph class~${\cal G}$ for which FO Model Checking is in $\FPT$ (here, ``{\it hereditary}'' means to be closed under vertex deletion). To explain this concept, we first need to introduce another notion, which we explain informally (see~\cite{BonnetC0K0T22} for a precise definition). Namely, a hereditary graph class ${\cal D}$ is {\it monadically independent} if it is possible to construct the class of all graphs from ${\cal D}$ using a $2$-stage process, which is called a {\it (first-order) transduction}. In the first stage, the vertices of every graph in ${\cal D}$ are non-deterministically coloured (that is, are assigned unary predicates). In the second stage, the coloured graphs are transformed to the class of all graphs by means of first-order formulas, using the colouring~\cite{BS1985monadic}. As FO Model Checking is $\AW$-hard on the class of all graphs~\cite{Downey96} (see also~\cite{FG}), 
a result of Dreier, M\"ahlmann, and Torunczyk~\cite{DreierMT24} implies that FO Model Checking is $\AW$-hard on every hereditary, monadically independent graph class. Now, a class of graphs $\mathcal{G}$ is said to be \emph{delineated} if for every hereditary subclass $\mathcal{D}$ of $\mathcal{G}$, the twin-width of $\mathcal{D}$ is bounded if and only if $\mathcal{D}$ is not monadically independent.  
We also recall that if an $n$-vertex graph $G$ has twin-width at most $t$, then deciding a first-order formula $\phi$ for $G$ can be done in time $f(t,|\phi|)\cdot O(n)$ for some computable function $f$ only depending on $t$ and the formula size $|\phi|$, assuming a $t$-contraction sequence is part of the input~\cite[Theorem~1.1]{BKTW22}. Hence, the following holds for every hereditary subclass ${\cal D}$ of a delineated graph class $\cal{G}$: if ${\cal D}$ has bounded twin-width, then FO Model Checking is in \FPT\ on ${\cal D}$ (if a suitable contraction sequence is given), otherwise it is $\AW$-hard.

It is known that the following graph classes are delineated: permutation graphs~\cite{BKTW22}, ordered graphs~\cite{twin-width4}, interval graphs or, more generally, rooted directed path graphs~\cite{BonnetC0K0T22}, and tournaments~\cite{GenietT23}. 
In Section~\ref{s-new}, we extend the delineation of the class of interval graphs, that is, $P_2$-graphs, to the class of $H$-graphs for every (simple) forest $H$, by refining the argument of Bonnet et al.~\cite{BonnetC0K0T22} for interval graphs.

\begin{restatable}{theorem}{thmdel}
\label{thm:effectiveDelineationHGraph}
For every forest $H$, the class of $H$-graphs is delineated.
\end{restatable}

\noindent
In Section~\ref{s-new}, we also show that for a forest $H$, it suffices to have an $H$-representation for input graphs from a hereditary subclass of $H$-graphs that has bounded twin-width.

\begin{restatable}{theorem}{twedo}\label{t-wedoabitmore}
Let $H$ be a forest and let ${\cal D}$ be a hereditary subclass of $H$-graphs such that every graph in ${\cal D}$ has twin-width at most~$t$, for some constant~$t$. It is possible to compute, in time 
$h(t)\cdot |V(G)|^{O(1)}$, an $f(t)$-contraction sequence of every graph $G\in {\cal D}$ from a given $H$-representation of $G$, where $h,f$ are computable functions that only depend on $t$.
\end{restatable}

\noindent
Combining Theorems~\ref{thm:effectiveDelineationHGraph} and~\ref{t-wedoabitmore} yields our first main result which, as mentioned, generalizes a known result for interval 
graphs~\cite{BonnetC0K0T22}.

\begin{restatable}{corollary}{thmdel2}\label{thm:effectiveDelineationHGraph2}
Let $H$ be a forest and let $\mathcal{D}$ be a hereditary subclass of $H$-graphs. If every graph in ${\cal D}$ has twin-width at most~$t$, for some constant~$t$, then FO Model Checking can be solved in time $f(t,\ell)\cdot |V(G)|^{O(1)}$, for some computable function $f$ only depending on $t$ and the formula size~$\ell$, provided that an
$H$-representation of the input graph $G\in \mathcal{D}$ is given. Otherwise, FO Model Checking is $\AW$-hard on ${\cal D}$ when parameterized by $\ell$.
\end{restatable}

\noindent
For a forest~$H$, we now apply \Cref{thm:effectiveDelineationHGraph2} to proper $H$-graphs and non-crossing $H$-graphs. We first consider proper $H$-graphs.
Because proper $K_{1,3}$-graphs have unbounded twin-width~\cite{Jedelsky2021thesis}, we
immediately 
obtain the following negative result as a consequence of \Cref{thm:effectiveDelineationHGraph2}.

\begin{corollary}\label{c-pp}
FO Model Checking is $\AW$-hard on proper $K_{1,3}$-graphs when parameterized by the formula size $\ell$.
\end{corollary}

\noindent
But what about non-crossing $H$-graphs?
As we will see, Corollary~\ref{c-pp} will be in line with our result that there exist proper $K_{1,3}$-graphs that are not non-crossing $T$-graphs for any tree~$T$ (see Section~\ref{s-preliminaries}). Namely, we will give an affirmative answer to our second research question about boundedness of twin-width of non-crossing $H$-graphs, not only for forests~$H$ but for {\it all} multigraphs~$H$. 

In order to do this, we first show in Section~\ref{sec:sketch} that for every multigraph~$H$, the class of non-crossing $H$-graphs has bounded proper mixed-thinness (Theorem~\ref{thm:top-prop}). This width parameter has been introduced by
Balab\'an, Hlinen\'y and Jedelsk\'y~\cite{BHJ-DM}, who also proved that proper mixed-thinness is a less powerful width parameter than twin-width. We refer to Section~\ref{s-preliminaries} for the definition of proper mixed-thinness. 

\begin{restatable}{theorem}{thmtopprop}
\label{thm:top-prop}
For a multigraph~$H$, every non-crossing $H$-graph $G$ has proper mixed-thinness at most $t=2^{\Vert H \Vert}(4^{\Vert H \Vert}-1)(\Vert H\Vert^2+1)+\Vert H \Vert$; moreover, a proper
$t$-mixed-thin representation can be obtained in polynomial time from a non-crossing $H$-representation of $G$.
\end{restatable}

\noindent
In contrast to Theorem~\ref{thm:top-prop}, proper $H$-graphs may have unbounded proper mixed-thinness, because they may have unbounded twin-width~\cite{Jedelsky2021thesis}. However,
we prove Theorem~\ref{thm:top-prop} by adapting a proof of a corresponding but incomparable result of Balab\'an, Hlinen\'y and Jedelsk\'y~\cite{BHJ-DM}, which shows that some proper $H$-graphs
{\it do have} bounded proper mixed-thinness, namely the class of proper $H$-graphs with a 
proper $H$-representation~${\cal S}$ in which each set induces a path in the corresponding subdivision $H'$ of $H$. We state their result in Section~\ref{sec:sketch} (as Theorem~\ref{thm:BHJ-DM}) before proving Theorem~\ref{thm:top-prop}. 
Moreover, in Section~\ref{s-notthecase}, we show that for connected graphs~$H$ we cannot replace ``proper mixed-thinness'' by ``(proper) thinness'' in Theorem~\ref{thm:top-prop} unless $H$ is a tree. In fact, the thinness bound that we show for classes of $H$-graphs where $H$ is a tree improves a known bound on linear mim-width for those classes~\cite{Algo-H-graphs}.

Balab\'an, Hlinen\'y and Jedelsk\'y~\cite{BHJ-DM} proved that for every integer $t\geq 1$, the twin-width of a graph $G$ with proper mixed-thinness $t$ is at most $9t$, and moreover a $9t$-contraction sequence of $G$ can be computed in polynomial time from a proper $t$-mixed-thin representation of~$G$. Recall also that Bonnet et al.~\cite{BKTW22} proved that deciding a first-order formula $\phi$ for an $n$-vertex graph $G$ with a given $d$-contraction sequence can be done in $f(d,|\phi|)\cdot O(n)$ time for some computable function $f$  only depending on $d$ and $|\phi|$. Combining these two results with Theorem~\ref{thm:top-prop} yields the following.

\begin{corollary}\label{t-fpt}
For a multigraph~$H$, FO Model Checking is in $\mathsf{FPT}$ for non-crossing $H$-graphs when parameterized by $\Vert H\Vert + \ell$, where $\ell$ is the size of a formula, provided that a non-crossing $H$-representation of the input graph~$G$ is given.
 \end{corollary}

\noindent
Note that \Cref{t-fpt} can be viewed as a more general result than what is obtained by combining
\Cref{thm:effectiveDelineationHGraph2} and \Cref{thm:top-prop}, because doing the latter requires $H$ to be a (simple) forest instead of an arbitrary multigraph. However,  
Corollary~\ref{thm:effectiveDelineationHGraph2} only requires an $H$-representation as part of the input, whereas Corollary~\ref{t-fpt} requires a non-crossing $H$-representation.

\begin{table}[t]
\centering
\begin{tabular}{| l l l l|}
 \hline
 Parameter & Graph class & Complexity & Reference \\ [0.5ex]
 \hline\hline
  $\Vert H \Vert$ & $H$-graphs & $\mathsf{XP}$ & \cite{ChaplickTVZ21}\\
   & proper $H$-graphs & $\mathsf{XP}$ &  \cite{ChaplickTVZ21} (from $H$-graphs) \\
    & non-crossing $H$-graphs &  $\mathsf{XP}$ &  \cite{ChaplickTVZ21} (from $H$-graphs)\\
   $\Vert H \Vert + k$ & $H$-graphs & $\mathsf{W}[1]$-hard & \cite{Algo-H-graphs}
   (strengthened by \Cref{hardness})\\
   & proper $H$-graphs & $\mathsf{W}[1]$-hard  & \Cref{hardness} \\
    & non-crossing $H$-graphs & $\mathsf{FPT}$  & Corollary~\ref{t-fpt}\\
\hline
\end{tabular}
\caption{Complexity of \textsc{Independent Set} on $H$-graphs, where $k$ denotes the size of a solution.
It is still open whether {\sc Independent Set} is in $\mathsf{FPT}$
for non-crossing $H$-graphs when parameterized by $\Vert H \Vert$ only.}
\label{table:1}
\end{table}

Corollary~\ref{t-fpt} applies, in particular, to classic problems on non-crossing $H$-graphs, such as {\sc Independent Set}, {\sc Clique} and {\sc Dominating Set}, which are all special cases of the FO Model Checking problem.
Note that \Cref{c-pp} does not necessarily exclude fixed-parameter tractable algorithms for any of these problems on proper $H$-graphs. However,
our next result, proven in Section~\ref{s-hard}, provides hard evidence, assuming $\mathsf{W}[1]\neq \mathsf{FPT}$. It gives, together with Corollary~\ref{t-fpt}, a positive answer to our first research question, which was originally posed by Chaplick~\cite{Chap-nc-paths}.

\begin{restatable}{theorem}{hardness}
\label{hardness}
{\sc Independent Set} is $\mathsf{W}[1]$-hard for proper $H$-graphs when parameterized by $\Vert H\Vert +k$, even if a proper $H$-representation of the input graph $G$  is given.
\end{restatable}

\noindent
Theorem~\ref{hardness} also strengthens the aforementioned  $\mathsf{W}[1]$-hardness result of Fomin, Golovach and Raymond~\cite{Algo-H-graphs} for general $H$-graphs, when parameterized by $\Vert H\Vert+k$; see also the summary in Table~\ref{table:1}.

From our results we conclude that, perhaps somewhat surprisingly, non-crossing $H$-graphs are indeed the ``correct'' generalization of proper interval graphs in the $H$-graph setting, at least from the viewpoint of parameterized complexity. Nevertheless, there are still many open structural and algorithmic problems resulting from our paper, which we list in Section~\ref{s-conclusion}.

\section{Preliminaries}\label{s-preliminaries}

\textbf{Basic graph terminology.} A class of graphs is {\it hereditary} if it is closed under vertex deletion. Note that for every graph $H$, the classes of $H$-graphs, proper $H$-graphs and non-crossing $H$-graphs are hereditary.

Let $G$ be a graph. For disjoint subsets $S, T \subseteq V(G)$, we say that $S$ is \textit{complete to} $T$ if every vertex of $S$ is adjacent to every vertex of $T$, whereas $S$ is \textit{anticomplete to} $T$ if there are no edges with one endpoint in $S$ and the other in $T$. A set $M\subseteq E(G)$ is a {\it matching} if no two edges of~$M$ have a common end-vertex. A matching~$M$ is {\it induced} if there is no edge in $G$ between vertices of different edges of $M$. An \textit{independent set} of $G$ is a set of pairwise non-adjacent vertices of $G$. The maximum size of an independent set of $G$ is denoted by $\alpha(G)$. A \textit{clique} of $G$ is a set of pairwise adjacent vertices of $G$. A set $S \subseteq V(G)$ \textit{induces} the subgraph $G[S] = (S, \{uv : u,v \in S, \ uv \in E(G)\})$.
Two vertices $u$ and $v$ in $G$ are \emph{true twins} if $u$ and $v$ are adjacent and have the same set of neighbours in $V(G)\setminus \{u,v\}$.

The complement of a graph $G$ is the graph $\overline{G}$ with vertex set $V(G)$, such that $uv \in E(\overline{G})$ if and only if $uv \notin E(G)$. For two vertex-disjoint graphs $G_1$ and $G_2$, the \emph{disjoint union} of $G_1$ and $G_2$ is the graph~$G$ with $V(G) = V(G_1) \cup V(G_2)$ and $E(G) = E(G_1) \cup E(G_2)$.

A {\it forest} is a graph with no cycles and a {\it tree} is a connected forest.
The path, cycle and complete graph on $n$ vertices are denoted by $P_n$, $C_n$ and $K_n$, respectively. The complete bipartite graph with partition classes of sizes $s$ and $t$ is denoted by $K_{s,t}$; recall that $K_{1,3}$ is also known as the claw.


\medskip
\noindent
\textbf{$H$-graph representations.} For a 
 multigraph $H$, let $G$ be an $H$-graph with $H$-representation~${\cal S}$. Thus, ${\cal S}$ is a family of connected vertex subsets of some subdivision $H'$ of $H$, in which case we call $H'$ the \emph{framework}.
Moreover, for every vertex $u\in V(G)$, there exists a unique vertex subset $S_u\subseteq V(H')$, which we call the {\it representative} of $u$ in ${\cal S}$. By definition, for every two distinct vertices $u$ and $v$ in $G$ it holds that $u$ and $v$ are adjacent if and only if $S_u\cap S_v\neq \varnothing$. Now, let $x$ and $y$ be two adjacent vertices of $H$. To obtain $H'$, the edge $xy$ may be replaced by a path $xz_1\cdots z_ry$ for some $r\geq 1$.
In that case, we say that $z_1,\ldots,z_r$ are the \emph{internal vertices} of the edge~$xy$.
If we order the path $xz_1\cdots z_ry$ from $x$ to $y$, then we obtain the ordered path $S(xy)$.

\medskip
\noindent
\textbf{Width parameters.} 
We first define the notion of twin-width, introduced in~\cite{BKTW22}. Afterwards, we will define the notion of thinness and several of its variants. 

A~\emph{trigraph} $G$ is a graph with two disjoint edge sets, the black edge set $E(G)$ and the red edge set $R(G)$.  
A~(vertex) \emph{contraction} consists of merging two (non-necessarily adjacent) vertices $u, v\in V(G)$ into a new vertex $w$, where the incidence of $w$ is defined as follows. There is a black edge $wz$ if and only if $uz$ and $vz$ were previously black edges. There is a red edge $wz$ if and only if either $uz$ or $vz$ was previously red or exactly one of the two pairs $uz, vz$ was a black edge (and the other was a non-edge). The rest of the graph remains unchanged.
A~\emph{contraction sequence} of an $n$-vertex (tri)graph $G$ is a sequence of trigraphs $G=G_n, \ldots, G_1=K_1$ such that $G_i$ is obtained from $G_{i+1}$ by performing one contraction.
A~\mbox{\emph{$d$-contraction sequence}} 
of $G$ is a contraction sequence $G=G_n, \ldots, G_1=K_1$ such that for every $i\in [n]$, every vertex $v\in V(G_i)$ is incident to at most~$d$ red edges.
The~\emph{twin-width} of $G$ is the minimum integer~$d$ such that $G$ admits a 
$d$-contraction sequence.

A graph $G=(V,E)$ is \emph{$k$-thin} if there is an ordering $v_1, \dots , v_n$ of~$V$ and a partition of~$V$ into $k$ classes $(V^1,\dots,V^k)$ such that for every triple $(r,s,t)$ with $r<s<t$ the following holds: if $v_r$ and $v_s$ belong to the same class $V^i$ for some $i\in\{1,\ldots,k\}$ and moreover $v_r v_t\in E$, then $v_sv_t \in E$. In this case, the ordering $v_1, \dots , v_n$ and the partition $(V^1,\dots,V^k)$ are said to be \emph{consistent} and form a {\it $k$-thin representation} of~$G$. The \emph{thinness} $\thin(G)$ of~$G$, introduced by Mannino et al.~\cite{M-O-R-C-thinness}, is the minimum integer~$k$ such that $G$ is $k$-thin.
See Figure~\ref{fig:thinness} for an example.
Interval graphs are known to be equivalent to $1$-thin graphs~\cite{Ola-interval,R-PR-interval}, and a vertex ordering consistent with the partition into a single class is called an \emph{interval order} for $G$.

\begin{figure}[t]
\begin{center}
    \begin{tikzpicture}[scale=0.7]

\foreach \x in {0.5} {
    \vertex{-6}{0+4}{v1};
    \vertex{-6}{3*\x+4}{v2};
    \vertex{-6}{6*\x+4}{v3};
    \vertex{-6}{8*\x+4}{v4};

    \vertex{-4}{2*\x+4}{w1};
    \vertex{-4}{5*\x+4}{w2};
    \vertex{-4}{10*\x+4}{w3};

    \vertex{-2}{1*\x+4}{z1};
    \vertex{-2}{4*\x+4}{z2};
    \vertex{-2}{7*\x+4}{z3};
    \vertex{-2}{9*\x+4}{z4};
    \vertex{-2}{11*\x+4}{z5};
}

\node[label=above:{$V^1$}] at (-6,2.5) {};
\node[label=above:{$V^2$}] at (-4,2.5) {};
\node[label=above:{$V^3$}] at (-2,2.5) {};

\path (v4) edge [e1,bend right=25] (v2); \path (v4) edge [e1]
(v3);

\path (w3) edge [e1] (w2);

\path (z4) edge [e1,bend left=25] (z2); \path (z4) edge [e1] (z3);
\path (z3) edge [e1] (z2);

\path (v2) edge [e1] (w1); \path (v3) edge [e1] (w2); \path (v4)
edge [e1] (w1); \path (v4) edge [e1] (w2); \path (v1) edge [e1]
(w2); \path (v4) edge [e1] (w3); \path (v2) edge [e1] (w2); \path
(v1) edge [e1] (w1);

\path (v1) edge [e1] (z1); \path (v2) edge [e1] (z2); \path (v2)
edge [e1,bend left=15] (z1); \path (v1) edge [e1,bend left=15]
(z2); \path (v2) edge [e1,bend left=15] (z3); \path (v4) edge [e1]
(z4); \path (v3) edge [e1] (z4); \path (v3) edge [e1] (z2); \path
(v3) edge [e1] (z3);

\path (z2) edge [e1] (w1); \path (z3) edge [e1] (w2); \path (z5)
edge [e1] (w3); \path (z5) edge [e1] (w2); \path (z2) edge [e1]
(w2); \path (z1) edge [e1] (w1); \path (z1) edge [e1] (w2);

\end{tikzpicture}
\end{center}
\vspace*{-0.75cm}
\caption{A $3$-thin representation of a graph. The vertices are ordered increasingly by their $y$-coordinate, and each class corresponds to a column of vertices.}\label{fig:thinness}
\end{figure}
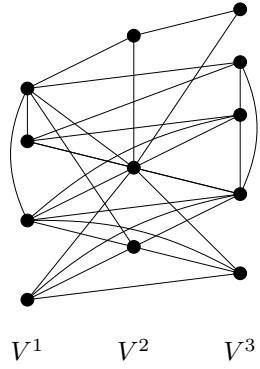

A graph $G = (V, E)$ is \textit{proper $k$-thin} if there is an ordering $v_1, \ldots, v_n$ of $V$ and a partition of $V$ into $k$ classes $(V^1, \ldots, V^k)$ such that for every triple $(r, s, t)$ with $r < s < t$ the following holds: if $v_r, v_s$ belong to the same class and $v_tv_r \in E$, then $v_tv_s \in E$ and if $v_s, v_t$ belong to the same class and $v_rv_t \in E$, then $v_rv_s \in E$. Equivalently, $G$ is proper $k$-thin if
both $v_1, \ldots, v_n$ and $v_n, \ldots, v_1$ are consistent with the partition. In this case, the partition and the ordering $v_1, \ldots, v_n$ are said to be \emph{strongly consistent}, and form a {\it proper $k$-thin representation} of $G$.  The \emph{proper thinness} $\pthin(G)$ of $G$, introduced by Bonomo and de Estrada~\cite{BE19}, is the minimum $k$ such that $G$ is proper $k$-thin and is denoted by $\pthin(G)$. Proper interval graphs are known to be equivalent to proper $1$-thin graphs~\cite{Rob-uig}, and a vertex ordering strongly consistent with the partition into a single class is called a \emph{proper interval order} for $G$.

We now define the parameters mixed-thinness and proper mixed-thinness. These parameters were introduced by Balab\'an, Hlinen\'y and Jedelsk\'y~\cite{BHJ-DM} and form a lower bound for thinness and proper thinness, respectively.
Let $G = (V, E)$ be a graph and let $k > 0$ be an integer.
For disjoint $W,W' \subseteq V$, we let $G[W,W']$ be the bipartite subgraph of $G$ induced by the edges with one end-vertex in $W$ and the other in $W'$. For convenience, we also let $G[W,W] = G[W]$. For two linear orders $<$ and $<'$ on the same set, we say that $<$ and $<'$ are \emph{aligned} if they are the same or one is the inverse of the other. Now, $G$ is \emph{$k$-mixed-thin} if there exists a partition $(V^1,\dots, V^k)$ of $V$ and, for each $1 \leq i \leq j \leq k$, a linear order $<_{ij}$ on $V^i \cup V^j$ and a choice of
$E_{i,j} \in \{E(G[V^i,V^j]), E(\overline{G}[V^i,V^j])\}$ such that, for every $1 \leq i \leq j \leq k$, the following hold:

\begin{description}
\item[(AL)] the restriction of $<_{ij}$ to $V^i$ (resp. to $V^j$) is aligned with $<_{ii}$ (resp. $<_{jj}$);
\item[(CO)] for every triple $u, v, w$ such that ($\{u, v\} \subseteq V^i$ and $w \in V^j$) or ($\{u, v\} \subseteq V^j$ and $w \in V^i$),
we have that if $u <_{ij} v <_{ij} w$ and $uw \in E_{i,j}$, then $vw \in E_{i,j}$.
\end{description}

\noindent
A $k$-mixed-thin graph $G$ is \emph{proper} if, in addition to (AL) and (CO), the following holds:

\begin{description}
\item[(SC)] for every triple $u, v, w$ such that ($\{v, w\} \subseteq V^i$ and $u \in V^j$) or ($\{v, w\} \subseteq V^j$ and $u \in V^i$),
we have that if $u <_{ij} v <_{ij} w$ and $uw \in E_{i,j}$, then $uv \in E_{i,j}$.
\end{description}

\noindent
A (proper) $k$-mixed-thin graph $G$ is \emph{inversion-free} if (AL) is replaced with the following:

\begin{description}
\item[(IN)] the restriction of $<_{ij}$ to $V^i$ (resp. to $V^j$) is equal to $<_{ii}$ (resp. $<_{jj}$).
\end{description}

\noindent
We can represent the edge sets choice by a matrix $R \in \{-1,1\}^{k\times k}$ such that, for every $1 \leq i \leq j \leq k$,
$R_{ij} = 1$ if $E_{i,j} = E(G[V^i,V^j])$ and $R_{ij} = -1$ if $E_{i,j} = E(\overline{G}[V^i,V^j])$.
We then call the $3$-tuple $((V^1,\dotsc,V^k),\{<_{i,j}\, : 1 \le i \le j \le k\},R)$ a \emph{(proper) $k$-mixed-thin representation} of $G$.

It follows from the definitions of (proper) thinness and (proper) mixed-thinness that the class of $k$-mixed-thin graphs is a superclass of the class of $k$-thin graphs, and that the same holds in the ``proper'' case (in particular, every proper interval graph is proper $1$-mixed-thin).
Indeed, a (proper) $k$-thin graph $G$ with vertex partition $(V^1,\dotsc,V^k)$ and vertex order $<$ can be represented as (proper) $k$-mixed-thin by the $3$-tuple $((V^1,\dotsc,V^k),\{<_{i,j}\, : 1 \le i \le j \le k\},R)$, where $<_{i,j}$ is simply the restriction of $<$ to $V^i \cup V^j$ for each pair of indices  $1 \le i \le j \le k$, and $R$ is the all-ones matrix. The inclusion is strict. For instance, the cycle $C_4 = v_1v_2v_3v_4$ is not an interval graph, and hence not a $1$-thin graph. However, $\overline{C_4}=2K_2$ is a proper interval graph. Hence, the $3$-tuple $((V^1),<_{1,1},R)$, where $V^1=\{v_1,v_2,v_3,v_4\}$, $v_1 <_{1,1} v_3 <_{1,1} v_2 <_{1,1} v_4$ and $R$ is the singleton matrix with entry $-1$, is a proper $1$-mixed-thin representation of $C_4$.

\medskip
\noindent
{\bf On the difference between proper and non-crossing $H$-graphs.}
We finish the section by arguing that 
proper $H$-graphs and non-crossing $H$-graphs form, in general, incomparable graph classes. 

We already observed in \Cref{s-intro} that proper $C_2$-graphs form a subclass of the class of non-crossing $C_2$-graphs. We recall that this inclusion is strict, because the claw $K_{1,3}$ is a non-crossing $C_2$-graph (see Figure~\ref{fig:4-fan-rep}) but not a proper $C_2$-graph. In fact, for every $n\geq 4$, the $(n+1)$-vertex star $K_{1,n}$ is also a non-crossing $C_2$-graph. However, it is not difficult to see that there exists no graph $H$, such that for every $n\geq 3$, $K_{1,n}$ is a proper $H$-graph. Hence, the class of non-crossing $C_2$-graphs is not contained in any class of proper $H$-graphs.

Similarly, it is also true that the class of proper $K_{1,3}$-graphs is not contained in any class of non-crossing $H$-graphs.
The reason is that proper $K_{1,3}$-graphs have unbounded twin-width~\cite{Jedelsky2021thesis}, whereas \Cref{thm:top-prop} states that for every multigraph $H$, the class of non-crossing $H$-graphs has bounded proper mixed-thinness, and thus bounded twin-width~\cite{BHJ-DM}. 

We also observe that the $4$-fan (i.e., $P_5$ plus a universal vertex) is a proper claw-graph. In fact, it can be represented as the intersection graph of a proper family of paths on a claw (see \Cref{fig:4-fan-rep}). However, to illustrate the difference between proper $H$-graphs and non-crossing $H$-graphs even more, we now show that the $4$-fan is not a non-crossing $T$-graph, for any tree $T$.

\begin{figure}[ht]
         \centering
        \begin{tikzpicture}[scale=.7]
        \draw [gray] (-2,0) -- (1.8,0);
        \draw [gray] (0,0) -- (0,-1.1);
        \draw  (-1.2,0.3) -- (1,0.3);
            \draw  (-1.8,0.1) -- (-1,0.1);
            \draw  (0.8,0.1) -- (1.6,0.1);
            \draw  (-0.1,0.2) -- (1.3,0.2);
            \draw  (-1.5,0.2) -- (-0.4,0.2);
            \draw  (-0.6,0.1) -- (0.1,0.1);
            \draw  (0.1,-0.9) -- (0.1,0.1);

\centerarc[gray](6,-0.5)(0:360:0.7);
\centerarc[black](6,-0.5)(0:360:0.8);
\centerarc[black](6,-0.5)(10:50:0.9);
\centerarc[black](6,-0.5)(70:110:0.9);
\centerarc[black](6,-0.5)(130:170:0.9);

        \end{tikzpicture}
         \caption{\footnotesize A proper claw-representation of the $4$-fan (left) and a non-crossing $C_2$-representation of the claw (right).}
        \label{fig:4-fan-rep}
     \end{figure}

\begin{observation}\label{obs:4-fan}
The $4$-fan is not a non-crossing $T$-graph, for any tree $T$.
\end{observation}

\begin{proof} Suppose we have a $4$-fan with vertices $v_1,v_2,v_3,v_4,v_5$ inducing a path and vertex $z$ adjacent to all of them. In particular, $v_1, v_3, v_5$
and $z$ induce a claw and, as it was proved in~\cite{Chap-nc-paths}, in every non-crossing $T$-representation where $T$ is a tree, $T_z$ is a subtree of a subdivision of $T$ with a vertex $x$ of degree at least
three, such that each of $T_{v_1}, T_{v_3}, T_{v_5}$ intersects $T_z$ in a different branch with respect to $x$. As $T_{v_2}$ intersects $T_{v_1}$ and $T_{v_3}$, it contains $x$; and as
$T_{v_4}$ intersects $T_{v_3}$ and $T_{v_5}$, it contains $x$ too, contradicting the fact that $T_{v_2}$ and $T_{v_4}$ are disjoint.
\end{proof}

\section{The Proofs of Theorems~\ref{thm:effectiveDelineationHGraph}
and~\ref{t-wedoabitmore}}\label{s-new}

For a set $U$ with total order $<$ and $A,B\subseteq U$, we write $A<B$ if $a<b$ for every $a\in A$ and $b\in B$. In the following all matrices are taken over $\mathbb{F}_2$.
A \emph{$k$-division} of a matrix $M$ consists of a partition $\mathcal{A}=(A_1,\dots, A_k)$ of the rows of $M$ and a partition $\mathcal{B}=(B_1,\dots, B_k)$ of the columns of $M$ such that $A_1\prec \dots \prec A_k$ and $B_1\prec \dots \prec B_k$ (and hence, all parts are consecutive in $\prec$). We denote the submatrix of~$M$ with rows in $A_i$ and columns in $B_j$ by $A_i\cap B_j$. We call such a submatrix a \emph{zone} of the $k$-division. A \emph{rank-$k$ division} of a matrix $M$ is a $k$-division $\mathcal{A}=(A_1,\dots, A_k)$, $\mathcal{B}=(B_1,\dots, B_k)$ such that every zone $A_i\cap B_j$ has at least $k$ distinct rows or $k$ distinct columns. Note that for large values of $k$, this implies that each zone of $M$ has rank at least $2$. The \emph{grid rank} of a matrix $M$, denoted $\gr(M)$, is the maximum $k\in \mathbb{N}$ such that $M$ admits a rank-$k$ division.
For a graph $G$ and a total order $\prec$ of the vertices of $G$, we let $\adj{\prec}{G}$ denote the adjacency matrix of $G$, for which rows and columns are ordered according to $\prec$.

We will use the following characterization of boundedness of twin-width.

\begin{theorem}[Bonnet et al.~\cite{twin-width4}]\label{thm:twwGridRank}
  There is a computable function $f\colon \mathbb N \to \mathbb N$ such that for every graph~$G$, the following two implications hold:\\[-10pt]
  \begin{description}
  \item [(i)] If $G$ has twin-width at most $k$,
  then there is a total order $\prec$ of $V(G)$ such that $\gr(\adj{\prec}{G}) \leqslant f(k)$;
  \item [(ii)] If there is a total order $\prec$ of $V(G)$ such that $\gr(\adj{\prec}{G}) \leqslant k$, then $G$ has
  twin-width at most $f(k)$.\\[-10pt]
  \end{description}
Moreover, there are computable functions $g, h\colon \mathbb N \to \mathbb N$ and an $h(k) \cdot |V(G)|^{O(1)}$-time algorithm that, for an adjacency matrix $\adj{\prec}{G}$ with no rank-$k$ division, outputs a $g(k)$-contraction sequence of~$G$.
\end{theorem}

We do not require the definition of delineation, but will make use of the following concept instead. 
For some integer~$k$, a {\it transversal pair $T_k$ (of half graphs)}
of a graph $G$ consists of three pairwise disjoint vertex subsets $A,B,C$ of $G$, each of size $k^2$. The edges of $T_k$ form a half-graph (also known as chain graph or ladder) between $A$ and $B$ and a half-graph between $B$ and $C$ such that the two vertex orderings of $B$ obtained by considering neighbourhood inclusion with respect to $A$ or $C$ encode a universal permutation (see~\cite{BonnetC0K0T22} for a precise definition). We say that $T_k$ is {\it semi-induced} if
we allow extra edges within $A$, $B$ or $C$, or between $A$ and $C$.

\begin{lemma}[Bonnet et al.~\cite{BonnetC0K0T22}]\label{lem:seqOrTransversalPair}
  Let $f\colon \mathbb N \to \mathbb N$ be any computable 
  function, and let $\mathcal C$ be a class of graphs.
  If for every $k\in \mathbb{N}$ and every $G \in \mathcal C$, either $G$ admits an  
  $f(k)$-contraction sequence or $G$ has a semi-induced transversal pair~$T_{k}$, then $\mathcal C$ is delineated.
\end{lemma}

For an arbitrary (simple) forest~$H$,
we will now reduce the case of $H$-graphs 
to the case of interval graphs. In interval graphs, having the configuration described in the following lemma is sufficient to find the obstruction needed to show delineation. For a fixed interval representation of a graph $G$ and vertex set $X \subseteq V(G)$, $\start(X)$ denotes the set of all start-points of the intervals representing vertices in $X$. Moreover, an interval representation is \textit{minimal} if the sum of the lengths of all intervals is minimal (while assuming integer start-points and end-points).

\begin{lemma}[Bonnet et al.~\cite{BonnetC0K0T22}]\label{lem:transversalPairInIntervalGraphs} Let $G$ be an interval graph and let $\adj{}{G}$ be an adjacency matrix of $G$. 
	Let $\mathcal{A}=\{A_1,\dots,A_{t^2}\}$ and $\mathcal{B}=\{B_1,\dots,B_{t^2}\}$ be two classes of pairwise disjoint vertex sets such that the following properties hold:
    
	\begin{enumerate}
		\item $\start(A_1) < \cdots  <\start(A_{t^2})\leq \start(B_1)< \cdots < \start(B_{t^2})$,
		\item The submatrix of $\adj{}{G}$ induced by rows $A_{i}$ and columns $B_{j}$ has rank at least $2$,
	\end{enumerate}
    
    \noindent where the starting intervals ${\sf start}(A_i)$ and ${\sf start}(B_j)$ are defined with respect to a minimal interval representation of $G$.  
	Then $G$ contains a semi-induced transversal pair $T_t$.
\end{lemma}

In order to use \cref{thm:twwGridRank}, we first define a suitable total ordering of the vertices of an $H$-graph, for which we rely on the $H$-representation of the graph. Next, we show in \cref{lem:findingTransversalPair} below that, given a graph for which the grid rank with respect to our chosen total order is large, we can find a path in $H$ along which we can find the representation of an interval graph satisfying the conditions in
\cref{lem:transversalPairInIntervalGraphs}.

We first define a total ordering for $H$-graphs in the case where $H$ is a tree. Afterwards, we will consider the more general case where $H$ is a forest.
Fix a tree $H$ and consider an $H$-graph $G$ with $H$-representation $\mathcal{S}$ and framework $H'$. We make the same minimality assumption as for interval graphs, that is, the sum of the cardinalities of all 
sets~$S_v$ is minimal. This assumption implies that for every vertex $x$ of $H'$ and adjacent vertices $u,v\in V(G)$ such that $x\in S_u$ but $x\notin S_v$, there is a vertex $w\in V(G)$ such that $x\in S_w$ and $w$ is not adjacent to $v$, that is, $w$ distinguishes $u$ and $v$. In particular, if $u$ and $v$ are true twins, the minimality assumption implies that $S_u=S_v$. 

We now fix an arbitrary root $r$ of $H$ and choose any linear order $e_1<\dots < e_{\Vert H \Vert}$ of the edges of $H$ such that $e_i<e_j$ if there is a root-to-leaf path traversing $e_i$ before $e_j$. We first define a linear order $\prec'$ on $V(H')$ as follows. To this end, we denote the path that replaces the edge $e_1$ in $H'$ by $z_1^{(1)}\cdots z_{\ell_1}^{(1)}$, where $r=z_1^{(1)}$. For every other edge $e_i \in E(H)$, we denote the path that replaces $e_i$ in $H'$ by $z_1^{(i)}\cdots z_{\ell_i}^{(i)}$, where the vertex closest to the root is omitted and $z_1^{(i)}$ is the vertex which is the second closest to the root $r$. The reason we exclude the first vertex for paths replacing edges $e_i$ ($i>2$) is so that every vertex in $H'$ gets only one label $z_i^{(j)}$.
For $z,z'\in V(H')$, let $e_i$ and $e_{i'}$ be the minimum edges (with respect to $<$) of $H$ such that $z\in \{z_1^{(i)},\dots, z_{\ell_i}^{(i)}\}$ and $z'\in \{z_1^{(i')},\dots, z_{\ell_{i'}}^{(i')}\}$. 
We let $z\prec' z'$ if
\begin{itemize}
    \item $i<i'$, or 
    \item $i=i'$ and, for some $j<k$, it holds that $z=z_j^{(i)}$ and $z'=z_k^{(i)}$.  
\end{itemize}

\noindent
We are now ready to define a linear order $\prec$ on $V(G)$ as follows. For $u,v\in V(G)$, we let $u\prec v$ if the minimum vertex (with respect to $\prec'$) in the symmetric difference $S_u\triangle S_v$ is contained in $S_u$; if $S_u=S_v$, we choose an arbitrary order of the two vertices.  
See \cref{fig:ordering} for an illustration of how to construct $\prec$.
We generalize this ordering $\prec$ to the case where $H$ is a forest by choosing an arbitrary ordering of the connected components of $H$ and concatenating the orderings $\prec$ for each connected component respecting the order of components that we have chosen.
Note that we can compute such a linear order $\prec$ in time $O(|V(G)|)$, given the $H$-representation of $G$, by following the steps above. 
\begin{figure}[t]
\begin{center}
    \begin{tikzpicture}[scale=0.44]
    \def \dist {6.5}

    \vertex{0}{0}{a};
    \vertex{-1}{-2}{b};
    \vertex{1.5}{-3}{c};
    \vertex{-3}{-6}{d};
    \vertex{0}{-4}{e};
    \vertex{-1}{-6}{f};
    \vertex{1}{-6}{g};

\node[label=above:{$H$ with $<$}] at (0,-9) {};
\node[label=above:{$H'$ with $\prec'$}] at (\dist,-9) {};
\node[label=above:{\textcolor{red}{$S_u$}}] at (2*\dist,-9) {};
\node[label=above:{\textcolor{blue}{$S_v$}}] at (3*\dist,-9) {};

\draw (a) -- (b) node [midway,fill=white,inner sep=0.02cm] {$e_1$};
\draw (a) -- (c) node [midway,fill=white,inner sep=0.02cm] {$e_2$};
\draw (b) -- (d) node [midway,fill=white,inner sep=0.02cm] {$e_3$};
\draw (b) -- (e) node [midway,fill=white,inner sep=0.02cm] {$e_4$};
\draw (e) -- (f) node [midway,fill=white,inner sep=0.02cm] {$e_5$};
\draw (e) -- (g) node [midway,fill=white,inner sep=0.02cm] {$e_6$};

    \vertex{0+\dist}{0}{aa};
    \vertex{-0.5+\dist}{-1}{aa1};
    \vertex{0.5+\dist}{-1}{aa2};
    \vertex{1+\dist}{-2}{aa3};
    \vertex{-1+\dist}{-2}{bb};
    \vertex{-1.5+\dist}{-3}{bb1};
    \vertex{-2+\dist}{-4}{bb2};
    \vertex{-2.5+\dist}{-5}{bb3};
    \vertex{-0.5+\dist}{-3}{bb4};
    \vertex{1.5+\dist}{-3}{cc};
    \vertex{-3+\dist}{-6}{dd};
    \vertex{0+\dist}{-4}{ee};
    \vertex{-0.5+\dist}{-5}{ee1};
    \vertex{0.5+\dist}{-5}{ee2};
    \vertex{-1+\dist}{-6}{ff};
    \vertex{1+\dist}{-6}{gg};

\vertexLabel[above]{aa}{$1$}
\vertexLabel[left]{aa1}{$2$}
\vertexLabel[left]{bb}{$3$}
\vertexLabel[right]{aa2}{$4$}
\vertexLabel[right]{aa3}{$5$}
\vertexLabel[right]{cc}{$6$}
\vertexLabel[left]{bb1}{$7$}
\vertexLabel[left]{bb2}{$8$}
\vertexLabel[left]{bb3}{$9$}
\vertexLabel[left]{dd}{$10$}
\vertexLabel[right]{bb4}{$11$}
\vertexLabel[right]{ee}{$12$}
\vertexLabel[left]{ee1}{$13$}
\vertexLabel[left]{ff}{$14$}
\vertexLabel[right]{ee2}{$15$}
\vertexLabel[right]{gg}{$16$}

\draw (aa) -- (bb);
\draw (aa) -- (cc);
\draw (bb) -- (dd);
\draw (bb) -- (ee);
\draw (ee) -- (ff);
\draw (ee) -- (gg);

    \vertex[gray!60]{0+2*\dist}{0}{aaa};
    \vertex[red]{-1+2*\dist}{-2}{bbb};
    \vertex[gray!60]{1.5+2*\dist}{-3}{ccc};
    \vertex[gray!60]{-3+2*\dist}{-6}{ddd};
    \vertex[red]{0+2*\dist}{-4}{eee};
    \vertex[gray!60]{-1+2*\dist}{-6}{fff};
    \vertex[gray!60]{1+2*\dist}{-6}{ggg};

\draw[gray!60] (aaa) -- (bbb);
\draw[gray!60] (aaa) -- (ccc);
\draw[gray!60] (bbb) -- (ddd);
\draw[gray!60] (bbb) -- (eee);
\draw[gray!60] (eee) -- (fff);
\draw[gray!60] (eee) -- (ggg);

    \vertex[red]{-0.5+2*\dist}{-1}{aaa1};
    \vertex[gray!60]{0.5+2*\dist}{-1}{aaa2};
    \vertex[gray!60]{1+2*\dist}{-2}{aaa3};
    \vertex[red]{-1.5+2*\dist}{-3}{bbb1};
    \vertex[red]{-2+2*\dist}{-4}{bbb2};
    \vertex[red]{-2.5+2*\dist}{-5}{bbb3};
    \vertex[red]{-0.5+2*\dist}{-3}{bbb4};
    \vertex[gray!60]{-0.5+2*\dist}{-5}{eee1};
    \vertex[red]{0.5+2*\dist}{-5}{eee2};

    \vertex[gray!60]{0+3*\dist}{0}{aaaa};
    \vertex[blue]{-1+3*\dist}{-2}{bbbb};
    \vertex[gray!60]{1.5+3*\dist}{-3}{cccc};
    \vertex[gray!60]{-3+3*\dist}{-6}{dddd};
    \vertex[blue]{0+3*\dist}{-4}{eeee};
    \vertex[blue]{-1+3*\dist}{-6}{ffff};
    \vertex[gray!60]{1+3*\dist}{-6}{gggg};

\draw[gray!60] (aaaa) -- (bbbb);
\draw[gray!60] (aaaa) -- (cccc);
\draw[gray!60] (bbbb) -- (dddd);
\draw[gray!60] (bbbb) -- (eeee);
\draw[gray!60] (eeee) -- (ffff);
\draw[gray!60] (eeee) -- (gggg);

    \vertex[blue]{-0.5+3*\dist}{-1}{aaaa1};
    \vertex[gray!60]{0.5+3*\dist}{-1}{aaaa2};
    \vertex[gray!60]{1+3*\dist}{-2}{aaaa3};
    \vertex[blue]{-1.5+3*\dist}{-3}{bbbb1};
    \vertex[gray!60]{-2+3*\dist}{-4}{bbbb2};
    \vertex[gray!60]{-2.5+3*\dist}{-5}{bbbb3};
    \vertex[blue]{-0.5+3*\dist}{-3}{bbbb4};
    \vertex[blue]{-0.5+3*\dist}{-5}{eeee1};
    \vertex[blue]{0.5+3*\dist}{-5}{eeee2};
\end{tikzpicture}
\end{center}
\caption{An example of how to obtain the order $\prec$, where $\textcolor{red}{u}\prec \textcolor{blue}{v}$ in the example given in the figure.}\label{fig:ordering}
\end{figure}

We further introduce the following notation. For a vertex $u\in V(G)$, the \emph{start edge} of $u$ is the minimum edge $e_i$ (with respect to $<$) such that the minimum $x\in S_u$ (with respect to $\prec'$) is contained in the subdivision of $e_i$. We observe that, by construction, in the linear order $\prec$ all vertices of $G$ with start edge~$e_1$ precede all vertices with start edge $e_2$ and so on. We let $\min_v^{(i)}$ and $\max_v^{(i)}$ be the minimum and maximum vertex (with respect to $\prec'$) from $\{z_1^{(i)},\dots, z_{\ell_i}^{(i)}\}$ which is contained in $S_v$, respectively, and set $\min_v^{(i)}=\max_v^{(i)}=\bot$ if $S_v\cap\{z_1^{(i)},\dots, z_{\ell_i}^{(i)}\}=\varnothing$. 
We make a straightforward observation.
\begin{observation}\label{obs:propOrder}
    For every two vertices $u,v\in V(G)$, both with start edge $e_i$, and which are not true twins, $u\prec v$ if
    \begin{itemize}
        \item $\min_u^{(i)}\prec' \min_v^{(i)}$, or
        \item there is $j>i$ such that $\min_u^{(i)}=\min_v^{(i)}$, $\max_u^{(k)}=\max_v^{(k)}$ for every $i<k<j$, and $\max_v^{(j)} \prec' \max_u^{(j)}$. 
    \end{itemize}
\end{observation}

We are now ready to state and prove our key lemma. 

\begin{lemma}\label{lem:findingTransversalPair}
Let $H$ be a forest, and let $h= \Vert H \Vert+1$ and $f(t)=2(h \cdot (2 t^2)^{h}+h-1)$. 
For every $t\in \mathbb{N}$, if $G$ is an $H$-graph such that $\gr(\adj{\prec}{G}) \geq f(t)$, then $G$ contains a semi-induced transversal pair $T_t$. 
\end{lemma}
\begin{proof}
    For ease of notation, we sometimes identify rows and columns with the vertices they are indexed by.
    Let $\mathcal{A}=(A_1,\dots, A_{f(t)})$, $\mathcal{B}=(B_1,\dots, B_{f(t)})$ be a rank-$f(t)$ division of $\adj{\prec}{G}$. We can view this $f(t)$-division as a $2$-division in which each zone has a rank-$(h \cdot (2 t^2)^{h}+h-1)$ division.
    Furthermore, among the four resulting zones at least one does not contain any entries from the main diagonal and hence the rows and columns of this submatrix are indexed by two disjoint vertex sets. We restrict to this zone and let $\mathcal{A}'=(A_1',\dots, A_{h \cdot (2 t^2)^{h}+h-1}')$, $\mathcal{B}'=(B_1',\dots, B_{h \cdot (2 t^2)^{h}+h-1}')$ be the rank-$(h \cdot (2 t^2)^{h}+h-1)$ division obtained by omitting the other row/column parts.

    Let $C_i$ be the set of vertices with start edge $e_i$. Observe that $(C_1,\dots, C_{h})$ is a partition of both the rows and columns of $\adj{\prec}{G}$ into consecutive parts (the reason for choosing $h= \Vert H \Vert+1$ will become evident below). Hence, at most $h-1$ of the row parts $A_1',\dots, A_{h \cdot (2 t^2)^{h}+h-1}'$ and of the column parts $B_1',\dots, B_{h \cdot (2 t^2)^{h}+h-1}'$ contain vertices from more than one of the sets $C_1,\dots, C_{h}$. Let $\mathcal{A}''=(A_1'', \dots, A_{h \cdot (2 t^2)^{h}}'')$, $\mathcal{B}''=(B_1'', \dots, B_{h \cdot (2 t^2)^{h}}'')$ be the $(h \cdot (2 t^2)^{h})$-division obtained from $\mathcal{A}'$, $\mathcal{B}'$ by restricting to parts that are fully contained in one of the sets $C_1,\dots, C_{h}$. Note that at this point we only consider one connected component of $H$ because we only consider one specific edge which resides in one component of $H$. 
    
    Next, we observe that there must be two (possibly equal) indices $i,i'$ such that the set $C_i$ contains at least $(2 t^2)^{h}$ of the row parts of $\mathcal{A}''$ and the set $C_{i'}$ contains at least $(2 t^2)^{h}$ of the column parts of $\mathcal{B}''$. Let $\mathcal{A}'''=(A_1''', \dots, A_{(2 t^2)^{h}}''')$, $\mathcal{B}'''=(B_1''', \dots, B_{(2 t^2)^{h}}''')$ be the restriction of $\mathcal{A}''$ and $\mathcal{B}''$ to parts fully contained in $C_i$ or $C_{i'}$, respectively. Without loss of generality, we may assume that $e_i\leq e_{i'
    }$.

    \begin{claim}\label{intervalinH-graph}There exist two (possibly equal) indices $j, j'$ with $i\leq j\leq j' \leq i'$, a $t^2$-division $\mathcal{A}^\circ =(A_1^\circ, \dots, A_{t^2}^\circ)$, $\mathcal{B}^\circ=(B_1^\circ, \dots, B_{t^2}^\circ)$ obtained from $\mathcal{A}'''$ and $\mathcal{B}'''$ by deleting some row/column parts, and pairwise disjoint connected vertex subsets $I_1,\dots, I_{t^2}$, $J_1, \dots, J_{t^2}$ of $H'$ satisfying the following: 
    
    \begin{enumerate}[(i)]
        \item\label{item1} for every $1 \leq k \leq t^2$, $I_k\subseteq \{z_1^{(j)},\dots, z_{\ell_j}^{(j)}\}$ and $J_k\subseteq \{z_1^{(j')},\dots,z_{\ell_{j'}}^{(j')} \}$,
        \item\label{item2} $I_1\prec' \dots \prec' I_{t^2}\prec' J_1\prec' \dots \prec' J_{t^2}$, 
        \item\label{item3} for every $1 \leq k \leq t^2$ and $b\in B_k^\circ$, $\min_b^{(j')}\in J_k$, 
        \item\label{item4} for every $1 \leq k \leq t^2$ and $a\in A_k^\circ$, if $j=i$ then $\min_a^{(j)}\in I_k$, and if $j>i$ then $\max_a^{(j)}\in I_k$.
    \end{enumerate}
    \end{claim}

    \begin{claimproof}
    We first find an index $j$ satisfying the properties above by repeating the steps described below. To simplify the inductive argument, we may assume that the minimum (with respect to $\prec'$) in $\bigcup_{1\leq k \leq (2 t^2)^{h}}\bigcup_{b\in B_k'''}S_b$ is the vertex  $z_1^{(i')}$. If this is not the case, we may subdivide $e_{i'}$ in $H$ once, obtaining the new edges $e_{i'}^1$ and $e_{i'}^2$, where $e_{i'}^1$ is closer to the root of $H$. The new graph $H$ has one extra edge which is accounted for in our choice of $h$. We obtain a new total edge ordering $<$ of the new graph $H$ by replacing $e_{i'}$ by $e_{i'}^1$ directly followed by $e_{i'}^2$. This edge ordering ensures that the ordering $\prec$ we obtain coincides with the original ordering.
    
    Set $m_a^{(j)}=\min_a^{(j)}$ in case $j=i$, and $m_a^{(j)}=\max_a^{(j)}$ in case $j>i$. For every $1 \leq k \leq (2 t^2)^{h}$, set $m_k^{(j)}$ to be the minimum (with respect to $\prec'$) of $m_{a}^{(j)}$ over all $a\in A_k'''$.  
    We start with $j=i$ and increase $j$ at every unsuccessful iteration by one. We assume that in the previous iteration we either have found a suitable index $j$ (Case 1 below), or obtained $(2t^2)^{h+i-j}$ row parts $A_1^{(j)},\dots, A_{(2t^2)^{h+i-j}}^{(j)}$ (initially, we set $A_k^{(i)}=A_k'''$ for every $1 \leq k \leq (2 t^2)^{h}$) such that for every pair $a,a'\in \bigcup_{1\leq k \leq (2t^2)^{h+i-j}} A_k^{(j)}$ we have $m_{a}^{(k)}=m_{a'}^{(k)}$ for every $k<j$ (Case 2 below). 
    
    We make the following observation.
    \begin{nitem}\label{claim-intervalOverlap}
        For every $1 \leq k \leq (2 t^2)^{h}$ and $a\in A_{k}^{(j)}$, it holds that $m_{k}^{(j)}\preceq' m_{a}^{(j)}\preceq' m_{k+1}^{(j)}$. 
    \end{nitem}
        To show \eqref{claim-intervalOverlap}, observe first that, by \cref{obs:propOrder}, $m_{k}^{(j)}\preceq' m_a^{(j)}$. To argue that  $m_a^{(j)}\preceq' m_{k+1}^{(j)}$, let $a'\in A_{k+1}^{(j)}$ be a vertex for which $m_{a'}^{(j)}=m_{k+1}^{(j)}$. Recall that $A_{k}^{(j)}\prec A_{k+1}^{(j)}$ (and hence $a\prec a'$) and by assumption $m_{a}^{(k)}=m_{a'}^{(k)}$ for every $k<j$. By \cref{obs:propOrder} this implies that $m_{a}^{(j)}\preceq' m_{a'}^{(j)}=m_{k+1}^{(j)}$.

    \bigskip    
    We now distinguish the following cases.
    
    \noindent \textbf{Case 1.} There are at least $2t^2$ pairwise distinct vertices among $m_1^{(j)}, \dots, m_{(2t^2)^{h+i-j}}^{(j)}$. 
    
    Let $m_1'\prec' \dots \prec' m_{2t^2}'$ be such vertices and, for every $1\leq k \leq 2t^2$, let $A_k^{\circ'} \in \{A_1^{(j)},\dots, A_{(2t^2)^{h+i-j}}^{(j)}\}$ be a row part for which the minimum (with respect to $\prec'$) of $m_a^{(j)}$ over all $a\in A_k^{\circ'}$ is $m_k'$. In this case, for every $1\leq k \leq t^2$, we set $I_k$ to be the set of vertices of the subpath of $z_1^{(j)}\cdots z_{\ell_j}^{(j)}$ from $m_{2k-1}'$ to $m_{2k}'$. By construction, the sets $I_1,\dots, I_{t^2}$ are pairwise disjoint, connected and satisfy that $I_1\prec' \dots \prec' I_{t^2}$. By \eqref{claim-intervalOverlap}, we further guarantee that $m_a^{(j)}\in I_{k}$ for every $a\in A_{2k-1}^{\circ'}$. Hence, by setting $A_k^{\circ}=A_{2k-1}^{\circ'}$, we obtain that our current choice of $j$ is suitable. 
    
    \noindent \textbf{Case 2.} There are at most $2t^2-1$ pairwise distinct vertices among $m_1^{(j)}, \dots, m_{(2t^2)^{h+i-j}}^{(j)}$.
    
    Hence, there must be $m\in \{m_1^{(j)}, \dots, m_{(2t^2)^{h+i-j}}^{(j)}\}$ which appears at least $(2t^2)^{h+i-(j+1)}+1$ times among $m_1^{(j)}, \dots, m_{(2t^2)^{h+i-j}}^{(j)}$. We let $A_1^{(j+1)}\prec \dots \prec A_{(2t^2)^{h+i-(j+1)}}^{(j+1)}$ be the first (with respect to $\prec$) $(2t^2)^{h+i-(j+1)}$ parts for which the minimum coincides with $m$. By \eqref{claim-intervalOverlap} this implies that $m_a^{(j)}=m_{a'}^{(j)}$ for every pair of vertices $a,a'\in \bigcup_{1\leq k \leq (2t^2)^{h+i-(j+1)}} A_k^{(j)}$, providing the row parts for the next iteration. In this case, we increase $j$ by $1$ and proceed with the next iteration.

    We now argue that this process terminates with $j< i'$. Indeed, if this is not the case, then we obtain $2t^2\geq 2$ row parts $A_1^{(i')},\dots, A_{2t^2}^{(i')}$ such that, for every pair $a,a'\in \bigcup_{1\leq k \leq 2t^2} A_k^{(i')}$, we have $m_{a}^{(k)}=m_{a'}^{(k)}$ for every $k<i'$. Note that there are some vertices  $a\in A_2^{(i')}$, $b\in B_2'''$ which are adjacent. Otherwise, the rank of $A_2^{(i')}\cap B_2'''$ is less than $2$, as such zone only contains $0$ entries. By \Cref{obs:propOrder}, it follows that $\max_a^{(i')}\preceq' \max_{a'}^{(i')}$ for every $a'\in A_1^{(i')}$, since $a'\prec a$. Additionally, by \Cref{obs:propOrder}, we have that $\min_{b'}^{(i')}\preceq' \min_{b}^{(i')}$ for every $b'\in B_2'''$, as $b'\prec b$.   But then every vertex in $A_1^{(i')}$ is adjacent to every vertex in $B_1'''$. Hence, the rank of $A_1^{(i')}\cap B_1'''$ is  less than $2$, a contradiction.

    We finally argue that setting $j'=i'$ suffices. We set $m_k$ to be the minimum (with respect to $\prec'$) of $\min_{b}^{(i')}$ over all $b\in B_k'''$. First, we argue that $m_k\not= m_{k+1}$. Similarly to \cref{obs:propOrder},  $m_k= m_{k+1}$ would imply that $\min_b^{(i')}=m_k$ for every $b\in B_k'''$. Since the starting edge $e_i$ of the $a'$s precedes $e_{i'}$, this implies that, for example, in the zone $A_1^\circ\cap B_k'''$ every row is either constant $0$ or constant $1$. Since this contradicts that $A_1^\circ\cap B_k'''$ has rank at least $2$, we get that $m_k\not= m_{k+1}$. We define the set $J_k$ to be the subpath of the path $z_{1}^{(i')}\cdots z_{\ell_{i'}}^{(i')}$  from $m_{2k-1}$ to $m_{2k}$. By our observation before, we know that the sets $J_k$ are pairwise vertex disjoint, connected and $I_{t^2}\prec' J_1 \prec' \dots \prec' J_{t^2}$. Similarly to \eqref{claim-intervalOverlap}, we can additionally show that $\min_b^{(i')}\in J_k$ for every vertex $b\in B_{2k-1}$, and hence we set $B_k^\circ=B_{2k-1}$.
    \end{claimproof}

\noindent
    Equipped with \Cref{intervalinH-graph}, we now consider the graph $\widehat{G}=G[\bigcup_{1\leq k \leq t^2} (A_k^\circ\cup  B_k^\circ)]$, for which we build an intersection model as follows. Choose $\widehat{H}=P_2$ and 
    let $P$ be the path from $e_{j}$ to $e_{j'}$ (including $e_j$ and $e_{j'}$). We choose the framework $\widehat{H}'$ to be $P$ after successively contracting edges $xy$ for which there is no vertex $v\in V(\widehat{G})$ for which $x\in S_v$ and $y\notin S_v$ (or the other way around). For $u\in \bigcup_{1\leq k \leq t^2} (A_k^\circ\cup  B_k^\circ)$, we let $\widehat{S}_u=S_u\cap V(\widehat{H}')$ and set $\widehat{\mathcal{S}}=\{S_u: u\in \bigcup_{1\leq k \leq t^2} (A_k^\circ\cup  B_k^\circ)\}$. It is easy to observe that $\widehat{\mathcal{S}}$ is a minimal 
    $P_2$-representation of $\widehat{G}$ and hence $\widehat{G}$ is an interval graph. Furthermore, (\ref{item3}) and (\ref{item4}) imply that $\start(A_k^\circ)\subseteq I_k$ and $\start(B_k^\circ)\subseteq J_k$ for every $1\leq k \leq t^2$ (note that, unless $j=i$, the framework $\widehat{H}'$ traverses $e_j$ largest to smallest vertex). Hence, by (\ref{item2}), we have that $\start(A_1^\circ)<\dots < \start(A_{t^2}^\circ)<\start(B_1^\circ)<\dots < \start(B_{t^2}^\circ)$. Since additionally $A_k^\circ \cap A_\ell^\circ$ has rank  at least $2$, we   
    can now use \cref{lem:transversalPairInIntervalGraphs} to complete the proof.
\end{proof}

We can finally show \Cref{thm:effectiveDelineationHGraph,t-wedoabitmore}. First, we immediately obtain \Cref{thm:effectiveDelineationHGraph} by combining \cref{lem:findingTransversalPair}, \cref{lem:seqOrTransversalPair} and \cref{thm:twwGridRank}. As for \cref{t-wedoabitmore} (restated below for convenience), we use \cref{lem:findingTransversalPair} combined with some standard results from the literature.

\twedo*

\begin{proof} Let $\prec$ be the linear order on the vertex set of an $H$-graph defined above. We first argue that $\gr(\adj{\prec}{G})$ is bounded by a constant $k$ for every $G\in \mathcal{D}$. If this is not the case then, by \Cref{lem:findingTransversalPair}, for every $c\in \mathbb{N}$ there exists a graph $G\in \mathcal{D}$ that contains a semi-induced transversal pair $T_c$. But then $\mathcal{D}$ is monadically independent, i.e., we can transduce the class of all graphs $\mathcal{G}$ from $\mathcal{D}$ (this is implied by \Cref{lem:seqOrTransversalPair}, and an explicit transduction is given in \cite[Lemma~13]{BonnetC0K0T22}). Since $\mathcal{D}$ has bounded twin-width and boundedness of twin-width is closed under applying first-order transductions \cite[Theorem~8.1]{BKTW22}, we conclude that the class of all graphs $\mathcal{G}$ must have bounded twin-width, a contradiction (see, e.g., \cite{BKTW22} were several classes of graphs of unbounded twin-width are given). 
An inspection of the transduction in \cite[Lemma~13]{BonnetC0K0T22} and the proof of \cite[Theorem~8.1]{BKTW22} reveals that the constant $k$ is computable and depends only on the twin-width $t$ of $\mathcal{D}$, hence $k=f(t)$ for some computable function~$f$.

We can compute $\prec$ in $O(|V(G)|)$-time using the $H$-representation of $G\in\mathcal{D}$, and use \cref{thm:twwGridRank} to compute a $g(f(t))$-contraction sequence of $G$ in time $h(t)\cdot |V(G)|^{O(1)}$, for some computable functions $g,h$. 
\end{proof}

\section{The Proof of Theorem~\ref{thm:top-prop}}\label{sec:sketch}

In this section we give the proof of \Cref{thm:top-prop}.  
As mentioned, we prove it by adapting a proof of a result of Balab\'an, Hlinen\'y and Jedelsk\'y~\cite{BHJ-DM} using a different approach (as we explain in detail later).

\begin{restatable}[Balab\'an, Hlinen\'y and Jedelsk\'y~\cite{BHJ-DM}]{theorem}{thmBHJ}
\label{thm:BHJ-DM}
    Let $H$ be a simple connected 
    graph\footnote{We remark that the proof of \Cref{thm:BHJ-DM} in \cite{BHJ-DM} in fact works for every multigraph $H$.}, and let $k$ be the sum of the number of paths and the number of cycles in $H$.
    Let $G$ be a proper $H$-graph, admitting a proper $H$-representation where each vertex set induces a path in the subdivision of $H$.
    Then $G$ is proper $||H||^2k$-mixed-thin, and a proper $||H||^2k$-mixed-thin representation of $G$
     can be computed in polynomial time from such an $H$-representation of $G$.
\end{restatable}

Note that \Cref{thm:top-prop} and \Cref{thm:BHJ-DM} address incomparable graph classes.
We will need an auxiliary result, \cref{lem:circ-arc-2-pmixed}, which states that a proper circular-arc graph is inversion-free proper $2$-mixed-thin, and describes how to obtain a representation.
Note, in contrast, that proper circular arc-graphs have unbounded thinness (see \Cref{t-nobound}).
The key idea behind \cref{lem:circ-arc-2-pmixed} appears in~\cite{BHJ-DM} (in the proof of Theorem~3.7, see cases~3 and~6) and we provide a proof below.

We also require a lemma that we use in the proofs of \Cref{lem:circ-arc-2-pmixed} and \Cref{thm:top-prop}.
In order to prove that a graph is proper $k$-mixed-thin, we frequently need to extend two proper interval orders to a linear ordering satisfying (CO) and (SC). The next lemma characterises when this is possible (see \Cref{fig:bip-pat-comp}).
We first need to give the definition of bipartite pattern from~\cite{B-B-lagos21-dam}. 

A \emph{bipartite trigraph} $T$ is a $5$-tuple $(A(T) \cup B(T),E(T),N(T),U(T))$, where $A(T) \cup B(T)$ is the vertex set and $E(T)$ (\emph{edges}), $N(T)$ (\emph{non-edges}) and $U(T)$ (\emph{undecided edges}) form a partition of the set $\{(a,b): a \in A(T), b \in B(T)\}$. A bipartite graph $G = (A(G) \cup B(G),E(G))$ is a \emph{realisation} of a trigraph $T$ if $A(G) = A(T)$, $B(G) = B(T)$, and $E(G) = E(T) \cup U'$, where $U' \subseteq U(T)$. An \emph{ordered bipartite graph} is a bipartite graph $G$ with a linear ordering of each of $A(G)$ and $B(G)$. We define the same for a bipartite trigraph, and call it a \emph{bipartite pattern}. We say that an ordered bipartite graph is a \emph{realisation} of a bipartite pattern if they share the same sets of vertices and respective linear orderings and the bipartite graph is a realisation of the bipartite trigraph. When, in an ordered bipartite graph, no ordered subgraph (with induced partition and orders) is the realisation of a given bipartite pattern, we say that the ordered bipartite graph \emph{avoids} the bipartite pattern.

\begin{lemma}[Bonomo-Braberman and Brito~\cite{B-B-lagos21-dam}]\label{lem:patterns}
    Let $G$ be a graph, let $\{V^1,V^2\}$ be a partition of $V(G)$, and let $<_1$ and $<_2$ be proper interval orders of $V^1$ and $V^2$ respectively.
    The linear orders $<_1$ and $<_2$ can be extended to a linear order $<$ of $V(G)$ such that $(\{V^1,V^2\}, <, \allowbreak\left(\begin{smallmatrix*}[r]1 & -1
\\ -1 & 1 \end{smallmatrix*}\right))$ satisfies conditions (CO) and (SC) if and only if $G[V^1,V^2]$, ordered according to $<_1$ and $<_2$, avoids the bipartite patterns $\overline{R_1}$, $\overline{R_2}$, $\overline{R_4}$ and $\overline{R_4'}$.
\end{lemma}

\begin{figure}
    \begin{center}
    \begin{tikzpicture}[scale=0.65]

\foreach \x in {0,...,4,7}
    \foreach \y in {0,1}
        \vertex{2*\x}{\y}{v\x\y};

\foreach \x in {5,6}
    \foreach \y in {1}
        \vertex{2*\x}{\y}{v\x\y};

\foreach \x in {4,7}
    \foreach \y in {2}
        \vertex{2*\x}{\y}{v\x\y};

\path (v00) edge [e1,dotted] (v11); \path (v01) edge [e1,dotted] (v10); \path
(v00) edge [e1] (v10);

\path (v20) edge [e1,dotted] (v31); \path (v21) edge [e1,dotted] (v30); \path
(v21) edge [e1] (v31);

\path (v40) edge [e1,dotted] (v51); \path (v42) edge [e1,dotted] (v51); \path
(v41) edge [e1] (v51);

\path (v70) edge [e1,dotted] (v61); \path (v72) edge [e1,dotted] (v61); \path
(v71) edge [e1] (v61);

\vertexLabel[left]{v00}{$a$}
\vertexLabel[left]{v01}{$b$}
\vertexLabel[right]{v10}{$c$}
\vertexLabel[right]{v11}{$d$}
\vertexLabel[left]{v20}{$b$}
\vertexLabel[left]{v21}{$a$}
\vertexLabel[right]{v30}{$d$}
\vertexLabel[right]{v31}{$c$}
\vertexLabel[left]{v40}{$a$}
\vertexLabel[left]{v41}{$b$}
\vertexLabel[left]{v42}{$c$}
\vertexLabel[right]{v51}{$d$}

\vertexLabel[right]{v70}{$a$}
\vertexLabel[right]{v71}{$b$}
\vertexLabel[right]{v72}{$c$}
\vertexLabel[left]{v61}{$d$}

    \node at (1,-1) {$\overline{R_1}$};
    \node at (5,-1) {$\overline{R_2}$};
    \node at (9,-1) {$\overline{R_4}$};
    \node at (13,-1) {$\overline{R_4'}$};

    \end{tikzpicture}
    \end{center}
    \caption{The bipartite patterns appearing in \cref{lem:patterns}, ordered from top to bottom. The solid lines denote compulsory edges and the dotted lines are compulsory non-edges in the pattern.}\label{fig:bip-pat-comp}
\end{figure}

Before stating \cref{lem:circ-arc-2-pmixed}, we need one more notion.
Let $G$ be a proper circular-arc graph.
We now describe a natural way to obtain a partition $(W,W')$ of $V(G)$, and total orders $<_{W}$ and $<_{W'}$ on $W$ and $W'$ respectively, from a representation of $G$.
So let $\mathcal{S}$ be a $C_2$-representation of $G$, with framework on vertex set $\{x_1,x_2,\dotsc,x_p\}$, where $x_1,x_2,\dotsc,x_p,x_1$ is a cycle.
Note that we do not require that the representation $\mathcal{S}$ is proper (in particular, it might be non-crossing but not proper).
    Each vertex $s_i \in V(G)$ has a representative in $\mathcal{S}$ that we denote by $S_i$.
    Let $W$ be the subset of vertices of $G$ given by $\{s_i : S_i \textrm{ contains $x_1$}\}$, and let $W'=V(G)\setminus W$. 
    For each $s_i \in W'$, we have $S_i = \{x_i,\dotsc,x_j\}$ for $1 < i \leq j < p+1$; we obtain an order $<_{W'}$ on $s_i \in W'$ by first ordering on $i$, and then ordering on $j$, where we break ties arbitrarily.
    For each $s_i \in W$, we have either $S_i = \{x_i,\dotsc,x_p,x_1,\dotsc,x_j\}$ with $j \leq i$, in which case we let $i' = i$; or $S_i = \{x_1,\dotsc,x_j\}$ for $j > 1$, in which case we let $i' = p+1$;
    we obtain an order $<_{W}$ on $s_i \in W$ by first ordering on $i'$, and then ordering on $j$, where we break ties arbitrarily.
    We call $((W,W'),<_{W},<_{W'})$ a \emph{natural bipartition ordering of $G$ from $\mathcal{S}$}.

We are now ready to prove \Cref{lem:circ-arc-2-pmixed}.

\begin{restatable}{lemma}{circarcpmixed}\label{lem:circ-arc-2-pmixed}
    Let $G$ be a proper circular-arc graph. Then $G$ is inversion-free proper $2$-mixed-thin.
    Moreover, for a natural bipartition ordering $((W,W'),<_{W},<_{W'})$ of $G$, there is a total ordering $<$ of $V(G)$ that extends both $<_{W}$ and $<_{W'}$ such that $((W,W'), \{<_{W},<,<_{W'}\}, \left(\begin{smallmatrix*}[r]1 & -1 \\ -1 & 1 \end{smallmatrix*}\right))$ is a proper $2$-mixed-thin representation of $G$ that satisfies (IN), (CO), and (SC).
\end{restatable}

\begin{proof}
    Let $((W,W'),<_{W},<_{W'})$ be a natural bipartition ordering of $G$.
    It is clear that $G[W]$ and $G[W']$ are proper interval graphs and that $<_W$ and $<_{W'}$ are proper interval orders for these graphs, respectively. 
    
    We will now show that $G[W,W']$, ordered according to $<_W$ and $<_{W'}$, avoids the bipartite patterns
    $\overline{R_1}$, $\overline{R_2}$, $\overline{R_4}$ and $\overline{R_4'}$
    of \Cref{fig:bip-pat-comp}. \Cref{lem:patterns} then implies that 
    $<_W$ and $<_{W'}$
    can be extended to a linear order $<$ of $V(G)$ such that $((W,W'), \{<_{W},<,<_{W'}\}, \left(\begin{smallmatrix*}[r]1 & -1 \\ -1 & 1 \end{smallmatrix*}\right))$ satisfies (CO) and (SC) and clearly (IN) as well. To this end, suppose to the contrary that $a,b,c,d \in V(G)$ realise in $G$ one of the bipartite patterns
    $\overline{R_1}$, $\overline{R_2}$, $\overline{R_4}$ or $\overline{R_4'}$.
    Let $S_a = \{i_a,\dotsc,j_a\}$,
    $S_b = \{i_b,\dotsc,j_b\}$,
    $S_c = \{i_c,\dotsc,j_c\}$, and
    $S_d = \{i_d,\dotsc,j_d\}$ be the corresponding sets in the $C_2$-representation.

    Suppose first that $a,b \in W$ and $c,d \in W'$ realise one of the bipartite patterns $\overline{R_1}$ or $\overline{R_2}$. Since $ad, bc \not \in E(G)$, we have $j_a < i_d \leq j_d < i_a$ and $j_b < i_c \leq j_c < i_b$. Since $ac \in E(G)$, either $i_c \leq j_a$ or $i_a \leq j_c$.
    Suppose $\overline{R_1}$ is realised. Then $j_a \leq j_b$ and $j_c \leq j_d$ by the definition of $<_W$ and since the representation is proper.
    Thus, if $i_c \leq j_a$ then $i_c \leq j_a \leq j_b$, a contradiction; whereas if $i_a \leq j_c$ then $i_a \leq j_c \leq j_d$, a contradiction.
    Now suppose $\overline{R_2}$ is realised. Then $i_b \le i_a$ and $i_d \le i_c$.
    Thus, if $i_c \leq j_a$, then $i_d \leq i_c \leq j_a$, a contradiction; otherwise, $i_b \leq i_a \leq j_c$, a contradiction as well.

Suppose now that $a,b,c \in W$ and $d \in W'$ realise the bipartite pattern $\overline{R_4}$. Since $ad, cd \not \in E(G)$, we have $j_c < i_d \leq j_d < i_a$. Since $bd \in E(G)$, either $i_d \leq j_b$ or $i_b \leq j_d$. By the definition of $<_W$, in the first case we have that $i_d \leq j_b \leq j_c$, a contradiction, and in the second case we have that $i_a \leq i_b \leq j_d$, a contradiction as well.

Suppose finally that $d \in W$ and $a,b,c \in W'$ realise the bipartite pattern $\overline{R'_4}$. Since $ad, cd \not \in E(G)$, $j_d < i_a \leq j_c < i_d$. Since $bd \in E(G)$, either $i_d \leq j_b$ or $i_b \leq j_d$. By the definition of $<_{W'}$, in the first case we have that $i_d \leq j_b \leq j_c$, a contradiction, and in the second case we have that $i_a \leq i_b \leq j_d$, a contradiction as well.
\end{proof}

\Cref{lem:circ-arc-2-pmixed} might be of independent interest. Indeed, Balab\'an, Hlinen\'y and Jedelsk\'y~\cite{BHJ-DM} showed that the class of inversion-free proper $k$-mixed-thin graphs is a transduction of the class of posets of width at most $f(k)$, for some quadratic function $f$, and that the classes of trees and $d$-dimensional grids (for fixed $d$) have bounded inversion-free proper mixed-thinness. 

We are finally ready to show \Cref{thm:top-prop}, which we restate below for convenience. The proof is inspired by the proof of \Cref{thm:BHJ-DM} by Balab{\'a}n, Hlinen{\'{y}} and Jedelsk{\'{y}}~\cite{BHJ-DM}. Their approach was the following.
Suppose that $G$ is the intersection graph of a proper family of paths in a subdivision $H'$ of $H$.  First, partition the vertices of $G$ according to the vertices and (subdivided) edges of $H$ involved in their representatives.
We can assume, by possibly choosing a ``finer'' subdivision $H'$, that no end of a representative in $H'$ is a vertex of $H$, i.e., that the ends of all paths ``lie inside'' the subdivided edges of $H$. Then, refine the partition according to the (at most two) edges of $H$ whose subdivisions contain the endpoints of the representative path. Let us say that these are the ``special'' edges of $H$ for the class. Finally, we can order each class in a strongly consistent way, and show that for every pair of classes, the respective orders can be extended to a strongly consistent common total order (by possibly complementing the edges between the classes). This last part is heavily based on the fact that the family of paths is proper, which allows a vertex order within each class to be defined using the order of the corresponding paths endpoints on each of the (at most two) special edges of $H$ for that class.

In our proof, we start in a similar way: based on a non-crossing $H$-representation of a graph $G$ with framework $H'$, we first partition the vertices of $G$ according to the vertices and (subdivided) edges of $H$ involved in their representatives. Note that, unlike for the endpoints of the paths, there can be more than two degree-$1$ vertices of representatives, so it is not straightforward how to define an order with the desired properties.
Our approach is to fix an ``ambassador'' of each class, and prove that, thanks to the non-crossing property, for every other vertex in the class, its representative and the representative of the ambassador differ on at most two (subdivided) edges of $H$. We refine the partition classes according to edges, and prove that, again thanks to the non-crossing property, we are able to order these refined classes in a strongly consistent way, and such that, for every pair of classes, the respective orders can be extended to a strongly consistent common total order (after possibly complementing the edges between the classes). For this last part, we simplify our arguments by using the bipartite pattern characterization of 2-thin graphs in \Cref{lem:patterns}.

In the proof of \Cref{thm:top-prop}, we will repeatedly make use of the following two observations. First, the class of non-crossing $P_2$-graphs coincides with the class of proper interval graphs. Second, the class of non-crossing $C_2$-graphs where no vertex is represented by every vertex of the subdivision of $C_2$ coincides with the class of proper circular-arc graphs.

\thmtopprop*

\begin{proof}
    Suppose that $G$ has a non-crossing $H$-representation $\mathcal{X}$ on framework $H'$, where each $v \in V(G)$ has a representative $X_v$ in $\mathcal{X}$ (so $X_v$ is a connected vertex subset of $H'$).
For each vertex $v \in V(G)$, let $E^1_v \subseteq E(H)$ be the set of edges $e$ of $H$ such that both endpoints of $e$ and all the internal vertices of $e$ in $H'$ belong to $X_v$;
let $E^2_v \subseteq E(H)$ be the set of edges $e$ of $H$ such that both endpoints of $e$ belong to $X_v$ but not all the internal vertices of $e$ in $H'$ belong to $X_v$;
let $E^3_v \subseteq E(H)$ be the set of edges $e$ of $H$ such that $X_v$ contains exactly one endpoint of $e$ (and possibly some internal vertices of $e$ in $H'$); let $E^4_v \subseteq E(H)$ be the set of edges $e$ of $H$ such that none of the endpoints of $e$ belongs to $X_v$ but at least one internal vertex of $e$ in $H'$ belongs to $X_v$; finally, let $F^3_v \subseteq E(H)\times V$ be the set of pairs $(e,x)$ such that $e \in E^3_v$ and $x$ is the endpoint of $e$ that belongs to $X_v$. We call $P_v = (E^1_v, E^2_v, F^3_v, E^4_v)$ the \emph{profile} of $v$. See Figure~\ref{fig:profiles} for an example.

\begin{figure}[t]
\begin{center}
    \begin{tikzpicture}[scale=0.5]

    \vertex{-3}{5}{a};
    \vertex{-3}{7}{b};
    \vertex{-5}{7}{c};
    \vertex{-7}{7}{d};
    \vertex{-7}{5}{e};
    \vertex{-5}{5}{f};

\vertexLabel[below right]{a}{$a$}
\vertexLabel[above right]{b}{$b$}
\vertexLabel[above]{c}{$c$}
\vertexLabel[above left]{d}{$d$}
\vertexLabel[below left]{e}{$e$}
\vertexLabel[below]{f}{$f$}

\node[label=above:{$H$}] at (-5,2.5) {};
\node[label=above:{$H'$}] at (3,2.5) {};

\path (a) edge [e1] (b);
\path (a) edge [e1] (e);
\path (f) edge [e1] (c);
\path (e) edge [e1] (d);
\path (b) edge [e1] (d);

    \vertex{6}{5}{aa};
    \vertex{6}{8}{bb};
    \vertex{3}{8}{cc};
    \vertex{0}{8}{dd};
    \vertex{0}{5}{ee};
    \vertex{3}{5}{ff};

\vertexLabel[below right]{aa}{$a$}
\vertexLabel[above right]{bb}{$b$}
\vertexLabel[above]{cc}{$c$}
\vertexLabel[above left]{dd}{$d$}
\vertexLabel[below left]{ee}{$e$}
\vertexLabel[below]{ff}{$f$}

\path (aa) edge [e1] (bb);
\path (aa) edge [e1] (ee);
\path (ff) edge [e1] (cc);
\path (ee) edge [e1] (dd);
\path (bb) edge [e1] (dd);

\foreach \x in {1,2,3} {
    \vertex{0.5*\x}{5}{ef\x};
    \vertex{\x}{8}{ec\x};
}

\foreach \x in {0,1,2,3,4} {
    \vertex[red]{0.75*\x+3}{8}{bc\x};
}

\vertex{3}{6.5}{fc1};

\foreach \x in {4,5,6,8,10,12} {
    \vertex[red]{0.5*\x}{5}{af\x};
}

\foreach \x in {1,2,3} {
    \vertex[blue]{0}{0.75*\x+5}{ed\x};
}

\end{tikzpicture}
\end{center}
\caption{A graph $H$ and a subdivision $H'$ of $H$. The profiles of vertices $v$ (red) and $w$ (blue), as defined in the proof of \Cref{thm:top-prop}, are:
$E^1_v = \{ab, bc, af\}, E^2_v = \{cf\}, E^3_v=\{cd,ef\}, F^3_v=\{(cd,c),(ef,f)\}$, $E^4_v = \varnothing$; $E^1_w = E^2_w = E^3_w = F^3_w = \varnothing$, and $E^4_w=\{de\}$.}\label{fig:profiles}
\end{figure}

Note that, since every $X_v$ is a connected vertex subset of $H'$, either $E^4_v = \varnothing$ or we have $|E^4_v|=1$ and $E^1_v = E^2_v = E^3_v = F^3_v = \varnothing$.
There are $m$ possible profiles for which $|E^4|=1$.
When $E^4 = \varnothing$, there are $4^m-1$ ways of distributing the edges of $H$ into $E^1,E^2,E^3$ and $E(H)\setminus (E^1 \cup E^2 \cup E^3)$ in such a way that $E^1 \cup E^2 \cup E^3$ is non-empty; for each such way, there are at most $2^{|E^3|} \leq 2^m$ possibilities for choosing the endpoint of each edge in $F^3$.
Therefore, there are at most $(4^m-1)2^m+m$ different profiles.
Consider now the partition $\mathcal{P} = (V[P_u] : u \in V(G)) = (V_1, V_2, \dots , V_k)$ of $V(G)$, where $V[P_u] = \{v\in V(G) : P_v = P_u\}$.

\begin{claim}
\label{claim-ncprop}
Let $v$ and $w$ be distinct vertices in $V_i$, for some $i \in \{1,2,\dotsc,k\}$.
Then $X_v \setminus X_w$ is either empty or a non-empty connected subset of the internal vertices of some edge $e_1$ of $E(H)$, and $X_w \setminus X_v$ is either empty or a non-empty connected subset of the internal vertices of some edge $e_2$ of $E(H)$, where neither $e_1$ nor $e_2$ belongs to $E^1_v =E^1_w$.
\end{claim}
\begin{claimproof}
Observe that $V(H) \cap X_v = V(H) \cap X_w$.
This implies that $X_v \setminus X_w$ and $X_w \setminus X_v$ are either empty or a subset of vertices added during the subdivision process.
Since the family is non-crossing, $X_v \setminus X_w$ is either empty or a non-empty connected subset of the internal vertices of some edge $e_1$ of $E(H)$, and $X_w \setminus X_v$ is either empty or a non-empty connected subset of the internal vertices of some edge $e_2$ of $E(H)$, where neither $e_1$ nor $e_2$ belongs to $E^1_v =E^1_w$, as required.
\end{claimproof}

\noindent
We use \cref{claim-ncprop} to further refine the partition classes of $\mathcal{P}$ corresponding to profiles where $E^4 = \varnothing$.
Let $i \in \{1,\dotsc,k\}$ such that $V_i$ is such a partition class in $\mathcal{P}$.
We fix a representative $v_i$ in $V_i$ and let $V_i' = \{w\in V(G) : X_w = X_{v_i}\}$.
In particular, $v_i \in V_i'$.
For $e \in E(H)$, we define $W_i(e)$ as the set of vertices $w \in V_i$ such that at least one of $X_{v_i} \setminus X_w$ and $X_w \setminus X_{v_i}$ is non-empty and both are subsets of the internal vertices of $e$.
For $e_1,e_2 \in E(H)$ with $e_1 \neq e_2$, we define $Z_i(e_1,e_2)$ as the set of vertices $w \in V_i$ such that $X_{v_i} \setminus X_w$ is a non-empty subset of the internal vertices of $e_1$ and $X_w \setminus X_{v_i}$ is a non-empty subset of the internal vertices of $e_2$.
Now, by \cref{claim-ncprop}, there is a partition of $V_i$ given by $\mathcal{V}_i = \{V_i'\} \cup \{W_i(e) : e \in E(H) \textrm{ and } W_i(e) \neq \varnothing\} \cup \{Z_i(e_1,e_2) : e_1,e_2 \in E(H) \textrm{, } e_1 \neq e_2 \textrm{ and } Z_i(e_1,e_2) \neq \varnothing \}$.
We refine the partition $\mathcal{P}$ to a partition $\mathcal{P}'$ of $V(G)$ by replacing each class $V_i=V[P_u]$ of $\mathcal{P}$ for which $E^4_u = \varnothing$ with $\mathcal{V}_i$.
Observe that the partition $\mathcal{P}'$ has at most $((4^m-1)2^m)(m(m-1)+m+1)+m=2^m(4^m-1)(m^2+1)+m$ partition classes.

For each pair $W, W'$ of partition classes in $\mathcal{P}'$, we now explain how to choose a linear order on $W \cup W'$ and an appropriate $E_{W,W'}$ satisfying (AL), (CO) and (SC).
For each partition class $W \in \mathcal{P}'$, we write $E^1_W$, $E^2_W$, $F^3_W$, and $E^4_W$ to refer to $E^1_w$, $E^2_w$, $F^3_w$, and $E^4_w$ (respectively) for any $w \in W$.
Moreover, we say that $e \in E(H)$ is a \emph{special edge} for $W$ if either $W = V_i$ and $E^4_W = \{e\}$, or $W = W_i(e)$, or $W=Z_i(e_1,e_2)$ and $e\in\{e_1,e_2\}$.
So, each class has at most two special edges.
Note that, for any $W \in \mathcal{P}'$, $G[W]$ is a proper interval graph and, in particular, $G[W]$ is complete when $E^4_W=\varnothing$.
Since every partition class induces a proper interval graph in $G$, we choose $E_{W,W} = E(G[W])$ for each part $W \in \mathcal{P}'$.
In particular, true twins in a class can be placed consecutively in any order and behave like a single vertex.

We will first define an order on each class $W$ of $\mathcal{P}'$, having properties as given by the following claim.

\begin{claim}\label{claim-order} For each class $W$ of $\mathcal{P}'$, there is a linear order satisfying the following conditions:
\begin{enumerate}
\item\label{it:1} If $E^4_W = \{xy\}$, then either the order or its reverse is an extension of the partial order of $W$ defined by the order of
the left endpoints of paths in $\{X_v \cap S(xy)\}_{v \in W}$, breaking ties if possible by the order of
their right endpoints. 
Moreover, this partial order is the same as the one defined by exchanging ``left'' and ``right'' above, and the reverse of the one we get by considering $S(yx)$ instead of $S(xy)$.
\item\label{it:2} For each $(xy,x) \in F^3_W$, either the order or its reverse is an extension of the partial order of $W$ defined by the order of
the right endpoints of paths in $\{X_v \cap S(xy)\}_{v \in W}$.
Moreover, if $W$ has two special edges in $E^3_W$, then for one of the edges this condition holds for the order, and for the other it holds for the reverse.
\item\label{it:4} For each $xy \in E^2_W$, either the order or its reverse is an extension of the partial order of $W$ defined by the order of
the right endpoints of paths in $\{X_v \cap S(xy)\}_{v \in W}$ that contain $x$, breaking ties if possible by the order of
the left endpoints of paths in $\{X_v \cap S(xy)\}_{v \in W}$ that contain $y$.
Moreover,
this partial order is the reverse of the one defined by exchanging $x$ and $y$ above. 
\end{enumerate}
\end{claim}

\begin{claimproof}
We distinguish cases according to the type of the partition class $W$.

\medskip
\noindent
\textbf{Case 1.} The class $W$ is a set $V_i$.

In this case, the sets $\{X_v\}_{v\in V_i}$ are non-crossing connected subsets of the internal vertices of an edge $e=xy$ of $H$.
In particular, $V_i$ induces a proper interval subgraph of $G$. We order the vertices by the order of the left endpoints of paths in $\{X_v \cap S(xy)\}_{v \in V_i}$, breaking ties, if possible, by the order of their right endpoints, and arbitrarily, otherwise.

\medskip
\noindent \textbf{Case 2.} The class $W$ is a set $V_i'$.

In this case, the sets $\{X_v\}_{v\in V_i'}$ are all equal, and we arbitrarily order the vertices.

\medskip
\noindent \textbf{Case 3.} The class $W$ is a set $W_i(e)$ with $(e,x) \in F^3_{v_i}$.

We order the vertices by the order of the right endpoints of paths in $\{X_v \cap S(xy)\}_{v \in W}$, breaking ties arbitrarily, since vertices that are indistinguishable by their right endpoints are true twins.

\medskip
\noindent \textbf{Case 4.} The class $W$ is a set $W_i(e)$ with $e = xy \in E^2_{v_i}$.

We order the vertices by the order of the right endpoints of paths in $\{X_v \cap S(xy)\}_{v \in W}$ that contain $x$, breaking ties, if possible, by the order of the left endpoints of paths in $\{X_v \cap S(xy)\}_{v \in W}$ that contain $y$, and arbitrarily, otherwise.

\medskip
\noindent \textbf{Case 5.} The class $W$ is a set $Z_i(e_1,e_2)$ with $(e_1,x) \in F^3_{v_i}$ and $(e_2,x') \in F^3_{v_i}$, where $e_1=xy$ and $e_2=x'y'$.

We order the vertices by the order of the right endpoints of paths starting at $x$ in $X_{v} \cap S(xy)$. If there are no ties, this is exactly the reverse of the order obtained by ordering according to the right endpoints of paths starting at $x'$ in $X_{v} \cap S(x'y')$, because the family is non-crossing. {So we break ties in the first order in a way that preserves this property (if there are still ties, then $X_{v} = X_{v'}$ for some $v \neq v'$ that are true twins in $G$, and so we break the tie arbitrarily)}.

\medskip
\noindent \textbf{Case 6.} The class $W$ is a set $Z_i(e_1,e_2)$ with $e_1 \in E^2_{v_i}$ and $(e_2,x_0) \in F^3_{v_i}$.

Let $x_0, x_1, \dots, x_p$ be the path $S(e_2)$ and let $y_0, y_1, \dots, y_{p'}$ be the path $S(e_1)$, so $e_2=x_0x_p$ and $e_1=y_0y_{p'}$.
Then $X_{v_i} \cap \{x_0, \ldots, x_p, y_0, \ldots, y_{p'}\} = \{x_0, \ldots, x_j,$ $y_0, \ldots, y_{j_1}, y_{j_2}, \dots, y_{p'}\}$, where $0 \leq j < p$ and $j_1 +1 < j_2$. By the definitions of $Z_i$, $E^2$, and $F^3$, and since the family is non-crossing, for every vertex $v \in W$, we have that $X_{v} \cap \{x_0, \ldots, x_p,$ $y_0, \ldots, y_{p'}\} = \{x_0, \ldots, x_{j(v)},$ $y_0, \ldots, y_{j_1(v)}, y_{j_2(v)}, \ldots, y_{p'}\}$, where $j < j(v) < p$, $j_1(v) +1 < j_2(v)$ and either $j_1(v) < j_1$ and $j_2(v) = j_2$, or $j_1(v) = j_1$ and $j_2(v) > j_2$.

Suppose now that there exist $v,w \in W$ such that $j(v) < j(w)$.  By symmetry we assume that $j_1(v) < j_1$ and $j_2(v) = j_2$. Then, since the family is non-crossing, $j_1(w) \leq j_1(v) < j_1$, from which $j_2(w) = j_2$. Suppose further that there exists a third vertex $z \in W$ such that $j_1(z) = j_1$ and $j_2(z) > j_2$. Then $X_z \setminus X_v$, $X_v \setminus X_z$, $X_z \setminus X_w$ and $X_w \setminus X_z$ all have non-empty intersection with the internal vertices of $e_1$. Since the family is non-crossing, none of these sets has non-empty intersection with the internal vertices of $e_2$, from which $j(z)=j(v)$ and $j(z)=j(w)$, contradicting the fact that $j(v) < j(w)$.

From the previous paragraph, we obtain that either $j(v) = j(w)$ for every $v, w \in W$, or $j_2(v) = j_2$ for every $v \in W$, or $j_1(v) = j_1$ for every $v \in W$. In the first case, we order the vertices with $j_2(v) = j_2$ by the order of $j_1(v)$, breaking ties arbitrarily, and then the vertices with $j_1(v) = j_1$ by the order of $j_2(v)$, breaking ties arbitrarily. In the second case, we order the vertices by the order of $j(v)$, breaking ties by the inverse of the order of $j_1(v)$, and arbitrarily if there still is a tie. In the third case, we order the vertices by the order of $j(v)$, breaking ties by the order of $j_2(v)$, and arbitrarily if there still is a tie. Note that, if $j(v) < j(w)$, then $j_1(v) \geq j_1(w)$ and $j_2(v) \leq j_2(w)$, and if $j_1(v) > j_1(w)$ or $j_2(v) < j_2(w)$, then $j(v) \leq j(w)$, as the family is non-crossing. Therefore, the order obtained extends the order of $j(v)$, and the inverse of the order of $j_1(v)$ or the order of $j_2(v)$, respectively.

\medskip
\noindent \textbf{Case 7.}
The class $W$ is a set $Z_i(e_1,e_2)$ with $e_2 \in E^2_{v_i}$ and $(e_1,x_0) \in F^3_{v_i}$.

The analysis of this case is similar to that of Case 6. Let $x_0, x_1, \dots, x_p$ be the path $S(e_1)$ in $H'$, and let $y_0, y_1, \dots, y_{p'}$ be the path $S(e_2)$ in $H'$ (where $e_1=x_0x_p$ and $e_2=y_0,y_{p'}$ are edges in $H$). Then $X_{v_i} \cap \{x_0, \ldots, x_p, y_0, \ldots, y_{p'}\} = \{x_0, \ldots, x_j, y_0, \ldots, y_{j_1},$ $y_{j_2}, \ldots, y_{p'}\}$, where $0 \leq j < p$ and $j_1 +1 < j_2$.
By the definitions of $Z_i$, $E^2$ and $F^3$, and since the family is non-crossing, for every vertex $v \in W$, we have that $X_{v} \cap \{x_0, \ldots, x_p, y_0, \ldots, y_{p'}\} = \{x_0, \ldots, x_{j(v)}, y_0, \ldots, y_{j_1(v)},$ $y_{j_2(v)}, \ldots, y_{p'}\}$, where $0 \leq j(v) < j$,
$j_1 +1 \leq j_1(v) +1 < j_2(v) \leq j_2$
and either $j_1(v) > j_1$ and $j_2(v) = j_2$, or $j_1(v) = j_1$ and $j_2(v) < j_2$,

Suppose that there exist $v,w \in W$ such that $j(v) < j(w)$.  We assume, by symmetry, that $j_1(w) > j_1$ and $j_2(w) = j_2$.
Then, since the family is non-crossing, $j_1(v) \geq j_1(w) > j_1$, from which $j_2(v) = j_2$. Suppose now that there exists a third vertex $z \in W$ such that $j_1(z) = j_1$ and $j_2(z) < j_2$. Then $X_z \setminus X_v$, $X_v \setminus X_z$, $X_z \setminus X_w$ and $X_w \setminus X_z$ all have non-empty intersection with the internal vertices of $e_1$. Since the family is non-crossing, none of these sets has non-empty intersection with the internal vertices of $e_2$, from which $j(z)=j(v)$ and $j(z)=j(w)$, contradicting the fact that $j(v) < j(w)$.

From the previous paragraph, we obtain that either $j(v) = j(w)$ for every $v, w \in W$, or $j_2(v) = j_2$ for every $v \in W$, or $j_1(v) = j_1$ for every $v \in W$. In the first case, we order the vertices with $j_1(v) = j_1$ by the order of $j_2(v)$, breaking ties arbitrarily, and then the vertices with $j_2(v) = j_2$ by the order of $j_1(v)$, breaking ties arbitrarily. In the second case, we order the vertices by the order of $j(v)$, breaking ties by the inverse of the order of $j_1(v)$, and arbitrarily if there still is a tie. In the third case, we order the vertices by the order of $j(v)$, breaking ties by the order of $j_2(v)$, and arbitrarily if there still is a tie. Note that, if $j(v) < j(w)$, then $j_1(v) \geq j_1(w)$ and $j_2(v) \leq j_2(w)$, and if $j_1(v) > j_1(w)$ or $j_2(v) < j_2(w)$, then $j(v) \leq j(w)$, as the family is non-crossing. Therefore, the order obtained extends the order of $j(v)$ and the inverse of the order of $j_1(v)$ or the order of $j_2(v)$, respectively.

\medskip
\noindent \textbf{Case 8.} The class $W$ is a set $Z_i(e_1,e_2)$ with $e_1, e_2 \in E^2_{v_i}$.

Let $x_0, x_1, \dots, x_p$ be $S(e_1)$ in $H'$ (where $x_0, x_p \in V(H)$) and let $y_0, y_1, \dots, y_{p'}$ be $S(e_2)$ in $H'$ (where $y_0, y_{p'} \in V(H)$). Then
$X_{v_i} \cap \{x_0, \dots, x_p, y_0, \ldots, y_{p'}\} = \{x_0, \ldots, x_{j_1},x_{j_2}, \ldots, x_p,$ $y_0, \ldots, y_{j_3}, y_{j_4}, \ldots, y_{p'}\}$, where $j_1 + 1 < j_2$ and $j_3 + 1 < j_4$.
Observe now that, by the definitions of $Z_i$ and $E^2$ and since the family is non-crossing, the following property holds: for every vertex $v\in W$, $X_{v} \cap \{x_0, \ldots, x_p, y_0, \ldots, y_{p'}\} = \{x_0, \ldots, x_{j_1(v)}, x_{j_2(v)}, \ldots, x_p,$ $y_0, \ldots, y_{j_3(v)},$ $y_{j_4(v)},$ $\dots, y_{p'}\}$,
where $0 \leq j_1(v) \leq j_1$, $j_2 \leq j_2(v) \leq p$,
$j_3 < j_3(v)+1 < j_4(v) \leq j_4$, either $j_1(v) < j_1$ and $j_2(v) = j_2$, or $j_1(v) = j_1$ and $j_2(v) > j_2$,
and either $j_3(v) > j_3$ and $j_4(v) = j_4$, or $j_3(v) = j_3$ and $j_4(v) < j_4$.

Suppose that there exist $v,w \in W$ such that $j_1(v) < j_1(w)$.  We assume, by symmetry, that $j_3(w) > j_3$ and $j_4(w) = j_4$. Then, since the family is non-crossing, $j_3(v) \geq j_3(w) > j_3$, from which $j_4(v) = j_4$. Suppose now that there exists a third vertex $z \in W$ such that $j_3(z) = j_3$ and $j_4(z) < j_4$. Then $X_z \setminus X_v$, $X_v \setminus X_z$, $X_z \setminus X_w$ and $X_w \setminus X_z$ all have non-empty intersection with the internal vertices of $e_2$. Since the family is non-crossing, none of these sets has non-empty intersection with the internal vertices of $e_1$, from which $j_1(z)=j_1(v)$ and $j_1(z)=j_1(w)$, contradicting the fact that $j_1(v) < j_1(w)$.

From the previous paragraph, we obtain that either $j_1(v) = j_1(w)$ for every $v, w \in W$, or $j_3(v) = j_3$ for every $v\in W$, or $j_4(v) = j_4$ for every $v\in W$. We now analyse these three cases separately by splitting each of them into a number of subcases.

Suppose that $j_1(v) = j_1(w)$ for every $v, w \in W$. Then either $j_1(v) < j_1$ for every $v \in W$, or $j_1(v) = j_1$ for every $v \in W$. Suppose first that $j_1(v) < j_1$ for every $v \in W$. Then $j_2(v) = j_2$ for every $v \in W$. We thus order the vertices with $j_3(v) = j_3$ by the order of $j_4(v)$, breaking ties arbitrarily, and then the
vertices with $j_4(v) = j_4$ by the order of $j_3(v)$, breaking ties arbitrarily. Suppose finally that $j_1(v) = j_1$ for every $v \in W$. By repeating the argument from the two previous paragraphs, we obtain that either $j_2(v) = j_2(w)$ for every $v, w \in W$ (where the common value is strictly greater than $j_2$), or $j_3(v) = j_3$ for every $v \in W$, or $j_4(v) = j_4$ for every $v \in W$. If the former holds, we again order the vertices with $j_3(v)=j_3$ by the order of $j_4(v)$, breaking ties arbitrarily, and then the vertices with $j_4(v) = j_4$ by the order of $j_3(v)$, breaking ties arbitrarily. The remaining two cases are addressed in the following paragraphs.

Suppose that $j_3(v) = j_3$ for every $v\in W$. Similar to the previous paragraphs, we have that either $j_4(v) = j_4(w)$ for every $v, w \in W$ (where the common value is strictly smaller than $j_4$), or $j_1(v) = j_1$ for every $v\in W$, or $j_2(v) = j_2$ for every $v\in W$.
If the former holds, we order the vertices with $j_2(v)=j_2$ by the order of $j_1(v)$, breaking ties arbitrarily, and then the vertices with $j_1(v) = j_1$ by the order of $j_2(v)$, breaking ties arbitrarily. If $j_1(v) = j_1$ for every $v\in W$, we order the vertices by the order of $j_2(v)$, breaking ties by the inverse of the order of $j_4(v)$, and arbitrarily if there still is a tie. If $j_2(v) = j_2$ for every $v\in W$, we order the vertices by the order of $j_1(v)$, breaking ties by the order of $j_4(v)$, and arbitrarily if there still is a tie. Note that, if $j_2(v) < j_2(w)$ then $j_4(v) \geq j_4(w)$, if $j_1(v) < j_1(w)$ then $j_4(v) \leq j_4(w)$,
and if $j_4(v) > j_4(w)$ then $j_2(v) \leq j_2(w)$ and $j_1(v) \geq j_1(w)$,
as the family is non-crossing. Therefore, in the last two cases, the order obtained extends the order of $j_2(v)$ and the inverse of the order of $j_4(v)$, and the order of $j_1(v)$ and the order of $j_4(v)$, respectively.

Suppose finally that $j_4(v) = j_4$ for every $v \in W$. Either $j_3(v) = j_3(w)$ for every $v, w \in W$ (where the common value is strictly greater than $j_3$), or $j_1(v) = j_1$ for every $v \in W$, or $j_2(v) = j_2$ for every $v \in W$.
If the former holds, we order the vertices with $j_2(v)=j_2$ by the order of $j_1(v)$, breaking ties arbitrarily, and then the vertices with $j_1(v) = j_1$ by the order of $j_2(v)$, breaking ties arbitrarily.
If $j_1(v) = j_1$ for every $v \in W$, we order the vertices by the order of $j_2(v)$, breaking ties by the order of $j_3(v)$, and arbitrarily if there still is a tie.
If $j_2(v) = j_2$ for every $v\in W$, we order the vertices by the order of $j_1(v)$, breaking ties by the inverse of the order of $j_3(v)$, and arbitrarily if there still is a tie. Note that, if $j_2(v) < j_2(w)$ then $j_3(v) \leq j_3(w)$, if $j_1(v) < j_1(w)$ then $j_3(v) \geq j_3(w)$,
and if $j_3(v) < j_3(w)$ then $j_2(v) \leq j_2(w)$ and $j_1(v) \geq j_1(w)$, as the family is non-crossing. Therefore, in the last two cases, the order obtained extends the order of $j_2(v)$ and the order of $j_3(v)$, and the order of $j_1(v)$ and the inverse of the order of $j_3(v)$, respectively.
\end{claimproof}

\noindent
{Observe that the orders on each class $W$ of $\mathcal{P}'$ defined according to \Cref{claim-order} together with the choice $E_{W,W} = E(G[W])$ satisfy (AL), (CO) and (SC), as each such $W$ induces a proper interval graph in $G$. Now, for each pair $W, W'$ of distinct classes of $\mathcal{P}'$, we explain how to combine the defined order or its reverse on each of $W$ and $W'$ into a linear order on $W \cup W'$ which, together with an appropriate choice of $E_{W, W'}$, satisfies (AL), (CO) and (SC). Clearly, we may assume that $W$ and $W'$ are neither complete nor anticomplete to each other, or else combining the orders is trivial.}

We first require some extra terminology.
For $v, w \in V(G)$ and $e \in E(H)$, we say that $v$ and $w$ \emph{intersect on $e$} if $X_v \cap X_w \cap S(e) \neq \varnothing$, and that they \emph{disagree on $e$} if either $X_v \setminus X_w$ or $X_w \setminus X_v$ is a non-empty subset of the internal vertices of $e$. For any two distinct classes $W, W'$ of $\mathcal{P}'$ and $e \in E(H)$, we say that $W$ and $W'$ \emph{intersect on $e$} if $w$ and $w'$ intersect on $e$ for some $w \in W$ and $w' \in W'$, and that they have \emph{non-trivial interaction} on $e$ if they intersect on $e$ and $e$ is a special edge for $W$ or $W'$.

Let $W$ and $W'$ be two distinct classes of $\mathcal{P}'$ which are neither complete nor anticomplete to each other. Then they intersect only on edges of $H$ for which they have non-trivial interaction. Moreover, $W$ and $W'$ can have non-trivial interaction on at most four edges. We distinguish two cases: firstly when one of $E^4_W$ and $E^4_{W'}$ is non-empty, and secondly when $E^4_W = E^4_{W'} = \varnothing$.

Suppose first that one of $E^4_W$ and $E^4_{W'}$ is non-empty, say without loss of generality $E^4_W = \{e\}$. Then $E^4_{W'} = \varnothing$, for otherwise either $W = W'$ or $W$ is anticomplete to $W'$. Now, $e \in E^1_{W'} \cup E^2_{W'} \cup E^3_{W'}$. If $e \in E^1_{W'}$, the classes $W$ and $W'$ are complete to each other. If $e \in E^3_{W'}$, then $W \cup W'$ induces a proper interval graph in $G$, and the concatenation of a suitable combination of the order of $W'$ or its reverse and the order of $W$ or its reverse (chosen according to \Cref{claim-order}) gives a proper interval order of $W \cup W'$. If $e \in E^2_{W'}$, then $W \cup W'$ induces a proper circular-arc graph in $G$, where the subdivided cycle is obtained by identifying the endpoints of the subdivision of $e$. In this case, we choose $E_{W,W'} = E(\overline{G}[W,W'])$ since, by \Cref{claim-order} (items~\ref{it:1} and~\ref{it:4}), we have a natural bipartition ordering of $G[W \cup W']$, so we can apply \Cref{lem:circ-arc-2-pmixed}. This concludes the proof for the case where one of $E^4_W$ and $E^4_{W'}$ is non-empty.

Suppose finally that $E^4_W = E^4_{W'} = \varnothing$. If the classes intersect on an edge of $E^1_W \cup E^2_W \cup E^1_{W'} \cup E^2_{W'}$, then they are complete to each other. The same holds if they intersect on an edge $e \in E^3_W \cup E^3_{W'}$ such that either $(e,x) \in F^3_W \cap F^3_{W'}$ for some $x \in V(G)$, or $e$ is
neither special for $W$ nor for $W'$.
Therefore, we assume that every edge $e$ on which $W$ and $W'$ intersect is such that $(e,x) \in E^3_W$, $(e,y) \in E^3_{W'}$ (where $e = xy$), and $e$ is special for at least one of $W$ and $W'$.

If $W$ and $W'$ intersect on exactly one edge, then $W \cup W'$ induces a co-bipartite proper interval graph in $G$. Moreover, the concatenation of the order of $W$ and the reverse of the order of $W'$ (chosen according to \Cref{claim-order}) results in a proper interval order of $W \cup W'$.

If $W$ and $W'$ intersect on exactly two edges $e = xy$ and $e'=x'y'$ such that $\{(e,x), (e',x')\} \subseteq F^3_W$, and $\{(e,y), (e',y')\} \subseteq F^3_{W'}$, then $W \cup W'$ induces a co-bipartite proper circular-arc graph in $G$, where the subdivided cycle is obtained by identifying $x$ with $x'$ and $y$ with $y'$ in the subdivisions of $e$ and $e'$. In this case, we make the choice $E_{W,W'} = E(\overline{G}[W,W'])$ since, by \Cref{claim-order} (item~\ref{it:2}), we have a natural bipartition ordering of $G[W \cup W']$, so we can apply \Cref{lem:circ-arc-2-pmixed}.

If $W$ and $W'$ intersect on exactly three edges then, since the three edges are special for either $W$ or $W'$ and each class has at most two special edges, we may assume that $e_2$ is special for $W$ and not for $W'$, $e_3$ is special for $W'$ and not for $W$, and $e_1$ can be special for both $W$ and $W'$ (it is special for at least one of them). We make the choices in such a way that the conditions of \Cref{claim-order} (item~\ref{it:2}) are satisfied in $e_1$ for the selected order of $W$ and for the reverse of the selected order of $W'$. So, by \Cref{claim-order} (item~\ref{it:2}), they are satisfied in $e_2$ for the reverse of the selected order of $W$, and in $e_3$ for the selected order of $W'$. In this case, $W \cup W'$ induces a co-bipartite proper circular-arc graph $G$, where a representation $\mathcal{S}$ as a non-crossing $C_2$-graph can be obtained in the following way.
Let $S(e_1)$ be the path $x_1, \dots, x_p$.  We let $D$ be a subdivision of $C_2$ with a cycle on $x_1,\dotsc,x_p,w',w,x_1$.
We obtain representatives for vertices as follows.
For a vertex $v \in W$ such that $X_v \cap S(e_1) = [x_1,x_i]$, the representative is $\{w',w,x_1, \dots, x_i\}$ if $X_v$ intersects $X_{v'}$ on $e_2$ for some (and thus every) $v'$ in $W'$, and $\{w,x_1, \dots, x_i\}$ otherwise. For a vertex $v' \in W'$ such that $X_{v'} \cap S(e_1) = [x_i,x_p]$, the representative is $\{x_i, \dots, x_p,w',w\}$ if $X_{v'}$ intersects $X_v$ on $e_3$ for some (and thus every) $v$ in $W$, and the path $x_i, \dots, x_p,w'$ otherwise.
It is easy to see that this gives a $C_2$-representation of $G[W \cup W']$ with framework $D$, and that it is non-crossing, since the original representation $\mathcal{X}$ was. In this case, we make the choice $E_{W, W'} = E(\overline{G}[W,W'])$ since the chosen orders give us a natural bipartition ordering, so we are under the hypothesis of \Cref{lem:circ-arc-2-pmixed}.

If $W$ and $W'$ intersect on exactly four edges, then each such edge is special for exactly one of $W$ and $W'$. For any possible orders $w_1, \dots, w_t$ of $W$ and $w'_1, \dots, w'_{t'}$ of $W'$ obtained according to \Cref{claim-order}, there exist $j_1, j_2, j_1', j_2'$ such that the only non-edges between $W$ and $W'$ are those between $\{w_{j_1}, \dots, w_{j_2}\}$ and $\{w'_{j_1'}, \dots, w'_{j_2'}\}$. We then make the choice $E_{W, W'} = E(\overline{G}[W,W'])$, since it is clear that $G[W, W']$ avoids the bipartite patterns of \Cref{fig:bip-pat-comp}, and so \Cref{lem:patterns} completes the proof.
\end{proof}

\section{A Thinness Dichotomy}\label{s-notthecase}

In this section we show that non-crossing $H$-graphs and also proper $H$-graphs have unbounded thinness if and only if $H$ contains a cycle $C_s$ for some $s\geq 2$. As having unbounded thinness implies having unbounded proper thinness by definition, 
we cannot strengthen \Cref{thm:top-prop} from proper mixed-thinness to proper thinness.

We first prove a new bound on the thinness of $H$-graphs if $H$ is tree.

\begin{restatable}{theorem}{thmtrees}
\label{thm:trees}
Let $G$ be an $H$-graph for some tree~$H$ that has exactly $\ell$ leaves, for some integer $\ell \geq 1$. Then $G$ has thinness at most $\max\{1,\ell-1\}$. Moreover, given
an $H$-representation of $G$, a vertex ordering with a consistent partition of size $\max\{1,\ell-1\}$ can be computed in $O(|V(G)|\cdot|V(H)|)$ time.
\end{restatable}

\begin{proof} If $H$ is trivial, then $G$ is complete, and the result easily follows. Otherwise, let $H$ be a tree with $\ell$ leaves, $\ell \geq 2$, and let $G$ be an $H$-graph. Let $H'$ be a subdivision of $H$ such that
each vertex of $G$ corresponds to a subtree of $H'$. Notice that $H'$ has $\ell$ leaves too, and
it can be seen that there exists such $H'$ with at most $|V(H))| + 4|V(G)||E(H)|$ vertices~\cite[page 3287]{ChaplickTVZ21}.

We will first order and partition $V(H')$ into $\ell-1$ classes. To do that, we
root $H'$ at an arbitrary leaf, order $V(H')$ by a postorder traversal
(i.e., concatenating the recursive postorder of the subtrees formed by each of the children of the root and their descendants, and adding at the end the root itself), and
assign the root to the first class. If a node has one child, then
the only child is assigned to the same class of its parent. If a
node has more than one child, then the child whose subtree is the first in the recursive postorder concatenation is assigned to the
same class of its parent, and each of the other children starts a
new class. Thus, we obtain $\ell-1$ classes.

We now order and partition the vertices of $G$
according to the previously defined order and partition of the
roots of their corresponding subtrees of $H'$, breaking ties
arbitrarily (two vertices of $G$ may correspond to subtrees
of $H'$ with the same root).

We now show that the order and the partition of $V(G)$ thus obtained are consistent.
Let $r < s < t$ be vertices of $G$, with $r$ and $s$ in the same
class and $rt \in E(G)$. Let $H_r$, $H_s$, $H_t$ be the
corresponding subtrees of $H'$. Since $H_r$ and $H_t$ intersect,
their union is a subtree $T$ of $H'$. The root of $T$ is
necessarily the root of at least one of $H_r$ and $H_t$, and since
we ordered the trees by postorder of their roots, it must be the
root of $H_t$. Since $H_r$ and $H_t$ intersect, the root of $H_r$
belongs to $H_t$ and so either the roots of $H_r$
and $H_t$ are the same or the root of $H_t$ is an ancestor of the
root of $H_r$.

Since $r$ and $s$ are in the same class and $r < s$, the way of defining
the order and partition of $V(G)$ implies that either the roots of $H_r$ and $H_s$ are the same or the
root of $H_s$ is an ancestor of the root of $H_r$. Since $s < t$, the definition of the order of $V(G)$ implies that the root of $H_s$ belongs to the
path joining the root of $H_r$ and the root of $H_t$, and so it
belongs to $H_t$. By the definition of $G$, we then have that $st \in E(G)$, as required.
\end{proof}

\begin{remark}
Mim-width and linear mim-width are well-known width parameters that we did not discuss in our paper. However, \Cref{thm:trees} has the following implication for linear mim-width. Bonomo and de Estrada~\cite{BE19} proved that for every graph $G$, the linear mim-width of $G$ is at most its thinness.  Consequently, \Cref{thm:trees} implies that the linear mim-width is at most $\max\{1,\ell-1\}$ if $G$ is an $H$-graph for some tree $H$ with exactly $\ell$ leaves. Fomin et al.~\cite{Algo-H-graphs} proved that every $H$-graph has linear mim-width at most
$\max\{1,2\Vert H\Vert|\}$.
Hence, \Cref{thm:trees} improves upon this result
for the special case where $H$ is a tree.
\end{remark}

We can now prove the following dichotomy.

\begin{restatable}{theorem}{tnobound}
\label{t-nobound}
For a multigraph $H$, the classes of proper $H$-graphs and non-crossing $H$-graphs have bounded thinness if and only if $H$ is a forest.
\end{restatable}

\begin{proof}
Let $H$ be a multigraph. Suppose that $H$ is a forest. Every $H$-graph $G$ is the disjoint union of connected graphs $G_1, \ldots, G_p$, each of which is an $H'$-graph for some connected component $H'$ of $H$. Moreover, $\thin(G) \leq \max_{i\in \{1,\ldots,p\}}\thin(G_i)$ (see, e.g., \cite{BE19}). Hence, we can apply \Cref{thm:trees} to each connected component of $G$. 

Suppose instead that $H$ is not a forest. Hence, $H$ has a cycle $C_s$ for some $s\geq 2$.
It is known that, for every $t \geq 1$, the complement of an
induced matching on $t$ edges has thinness $t$~\cite{BGOSS-thin-oper,C-M-O-thinness-man}. We now show that
these graphs are $C_k$-graphs, for any $k \geq 2$. Let $t \geq 1$, and subdivide the
edges of $C_k$ in order to obtain $C_{2k'}$ with $k' \geq t$. It is easy to see that
the intersection graph of all the distinct paths on $k'$ vertices
of $C_{2k'}$ is the complement of an induced matching on $k'$ edges,
since each path on $k'$ vertices of $C_{2k'}$ is disjoint from
exactly one other path on $k'$ vertices of $C_{2k'}$.
Moreover, the family is both proper and non-crossing. This shows that for every $t \geq 1$ and every $k \geq 2$, the complement of an induced matching on $t$ edges is a proper and non-crossing $C_k$-graph.
\end{proof}

\section{The Proof of Theorem~\ref{hardness}}\label{s-hard}

In this section we prove our $\mathsf{W}[1]$-hardness result for {\sc Independent Set} on proper $H$-graphs parameterized by $\Vert H\Vert + k$. To put our result in perspective, we start with a small remark.

\begin{remark}
It is not true that for every $H$-graph $G$ there exists a function $f$ such that $G$ is a proper $H'$-graph for some graph $H'$ with $\Vert H' \Vert = f(\Vert H \Vert)$, even in the case $H=P_2$.
Indeed, as the claw is not a proper interval graph, every proper $H'$-representation of a claw has at least one vertex subset containing an original vertex of $H'$. It is then enough to consider the
$P_2$-graph (interval graph) that is the disjoint union of an arbitrarily large number of claws.
\end{remark}

\noindent
The following result is crucial for the proof.

\begin{proposition}\label{thm:proper-H}
Given an $H$-graph $G$ together with an $H$-representation of $G$, it is possible to find, in time polynomial in $|V(G)| + |V(H)|$, a parameter $k_1$ such that $k_1 \leq \alpha(G)$, a graph $H'$ with
$|V(H')| = 2|V(H)|+4k_1$, and a proper $H'$-representation of $G$.
\end{proposition}

\begin{proof} Without loss of generality, we will consider representations in which no two vertices are represented by exactly the same set. Otherwise, we can keep only one such vertex, say $v$, and add afterwards the necessary copies of $v$, all of them being represented by the new representative set of $v$.

Let $\mathcal{S}$ be an $H$-representation of $G$ with framework $H''$, where the representative of each $v \in V(G)$ is $S_v$. If we remove from $G$ the vertices whose representatives contain a vertex of $H$, the remaining vertices are exactly those whose representatives are intervals of the subdivision of an edge of $H$ not containing any vertex of $H$. For each of these edges of $H$, we traverse the edge in some arbitrary direction and compute a maximum set of pairwise disjoint intervals using the following greedy algorithm: Start with the empty set and iteratively add, among those intervals that do not intersect the current set, the one that ends first, breaking ties with the one that starts last. Observe that, by construction, every unpicked interval contains at least one of the endpoints of some picked interval. The union $I$ over the edges of $H$ of these sets of pairwise disjoint intervals clearly corresponds to an independent set of $G$ and so, if $|I| = k_1$, then $k_1 \leq \alpha(G)$.

Let now $W$ be the union of $V(H)$ and the $2k_1$ vertices of $H''$ corresponding to the endpoints of the intervals in $I$. Then, each representative of a vertex of $G$ in the original representation $\mathcal{S}$ contains at least one vertex of $W$. Moreover, if in the representation $\mathcal{S}$ the representative of the vertex $v$ is properly contained in that of $v'$, then both contain a common vertex in $W$ (one of the vertices in $W$ of the representative of $v$).

We now construct $H'$ and a proper $H'$-representation of $G$. First, take $H$ and subdivide its edges by adding precisely the $2k_1$ vertices corresponding to the endpoints of the intervals in $I$.
That is, we initially have that $V(H') = W$ and $H''$ is a subdivision of $H'$. Now, for each vertex $w \in W$, we consider the set of vertices of $G$ whose representatives contain $w$.
The inclusion of representatives defines a partial order on this set, and we can order them by extending such partial order. In $H'$ we add a leaf $w'$ attached to $w$, and extend the representatives of the vertices of $G$ that contain $w$ onto the edge $ww'$ (via a subdivision of $ww'$) by making them end up in reverse order to that of inclusion; in this way we avoid representatives properly included in others among the vertices whose representatives contain $w$. Clearly, $|V(H')| = 2|V(H)|+4k_1$. Moreover, as observed earlier, if the representative of a vertex $v$ is properly included in that of $v'$ in the original representation $\mathcal{S}$, then both contain a common vertex in $W$, and so the new representation of $G$ as an $H'$-graph is indeed proper. (Recall that we assumed that no two vertices are originally represented by the same set.)
\end{proof}

\noindent
We are now ready to prove \Cref{hardness}, which we restate below.

\hardness*

\begin{proof}
We provide a parameterized reduction from \textsc{Independent Set} on $H$-graphs, which is
$\mathsf{W}[1]$-hard parameterized by $\Vert H\Vert + k$, even if a representation of $G$ as an $H$-graph is given~\cite{Algo-H-graphs}. Let $(G, k)$ be an instance of this problem, where $G$ is given together with an $H$-representation and $k \in \mathbb{N}$. Recall that the question is whether $G$ has an independent set of size at least $k$.

By \Cref{thm:proper-H}, we find, in time polynomial in $|V(G)| + |V(H)|$, a parameter $k_1$ such that $k_1 \leq \alpha(G)$, a graph $H'$ with $|V(H')| = 2|V(H)|+4k_1$, and a proper $H'$-representation of $G$. If $k_1 \geq k$, then $(G, k)$ is a yes-instance. Otherwise, we have a representation of $G$ as a proper $H'$-graph with $|V(H')| < 2|V(H)|+4k$, and we ask whether $G$ has an independent set of size at least $k$. The conclusion follows from the fact that $|E(H')| + k_1 = (|E(H)| + 2|V(H)| + 4k_1) + k_1 < |E(H)| + 2|V(H)| + 5k \leq g(|E(H)| + k)$, for some function~$g$.
\end{proof}

\section{Conclusions}\label{s-conclusion}

We first proved that for every (simple) forest~$H$, the class of $H$-graphs is delineated. This generalizes a known result for interval graphs~\cite{BonnetC0K0T22}. As a consequence, we find that for
every hereditary subclass $\mathcal{D}$ of $H$-graphs, FO Model Checking is in \FPT\ if ${\mathcal D}$ has bounded twin-width and $\AW$-hard otherwise.
Hence, in particular, FO Model Checking is $\AW$-hard for proper $K_{1,3}$-graphs. We showed that this is in contrast to the situation for non-crossing $H$-graphs. Namely, we proved that even for every multigraph~$H$, the FO Model Checking problem is in $\mathsf{FPT}$ for non-crossing $H$-graphs when parameterized by $\Vert H \Vert+\ell$, where $\ell$ is the size of a formula. We did this by proving that for every 
multigraph~$H$, the class of non-crossing $H$-graphs has bounded proper mixed-thinness and thus bounded twin-width. 
Moreover, we strengthened a result of Fomin, Golovach and Raymond~\cite{Algo-H-graphs} by proving that a special case of the FO Model Checking problem, namely {\sc Independent Set}, is $\mathsf{W}[1]$-hard even on proper $H$-graphs when parameterized by $\Vert H \Vert +k$, where $k$ is the solution size. Hence, we answered, in two different ways, a recent question of Chaplick~\cite{Chap-nc-paths} about the difference between proper $H$-graphs and non-crossing $H$-graphs. We conclude that, 
among these two generalizations of proper interval graphs, the
non-crossing $H$-graphs have computational advantages over the proper $H$-graphs.

We finish our paper with a number of open problems resulting from our work. First, it would be good to increase our insights in recognizing non-crossing $H$-graphs. Recall that recognizing $H$-graphs or proper $H$-graphs is \NP-complete for certain graphs $H$. However, we do not know any graph $H$ for which recognizing non-crossing $H$-graphs is \NP-complete and whether boundedness of proper mixed-thinness would be of help.

\begin{open}
Is recognizing non-crossing $H$-graphs \NP-complete for some graph $H$?
\end{open}

\noindent
Note that, unless $\mathsf{P} = \mathsf{PSPACE}$, FO Model Checking is not in $\mathsf{FPT}$ for non-crossing $H$-graphs when parameterized by $\Vert H\Vert$ only, as it is $\mathsf{PSPACE}$-complete on any class of structures that contains at least one structure with at least two elements (see, e.g., \cite{FG}).
However, we leave as an open problem to determine whether {\sc Independent Set} and {\sc Clique} are in $\mathsf{FPT}$ for non-crossing $H$-graphs even when parameterized by $\Vert H\Vert$ only. Recall that Chaplick et al.~\cite{ChaplickTVZ21} proved that {\sc Clique} is para-$\mathsf{NP}$-hard on $H$-graphs when parameterized by $\Vert H\Vert$, while the problem is also still open for proper $H$-graphs.

\begin{open}
Is {\sc Independent Set} in $\mathsf{FPT}$ for non-crossing $H$-graphs when parameterized by $\Vert H \Vert$ only?
\end{open}

\begin{open}
Determine the complexity of {\sc Clique} for proper $H$-graphs and non-crossing $H$-graphs when parameterized by $\Vert H\Vert$.
\end{open}

\noindent
We also wonder whether there are other problems expressible in first-order logic that exhibit the same behaviour as {\sc Independent Set} on proper $H$-graphs.
Note that {\sc Dominating Set} is captured by the FO Model Checking framework. Hence, it is in $\mathsf{FPT}$ for non-crossing $H$-graphs when parameterized by $\Vert H \Vert+k$ due to Corollary~\ref{t-fpt}.
However, our technique for proving $\mathsf{W}[1]$-hardness for {\sc Independent Set} does not work for {\sc Dominating Set}.

\begin{open}
Is {\sc Dominating Set} $\mathsf{W}[1]$-hard for proper $H$-graphs when parameterized by $\Vert H\Vert+k$?
\end{open}

\noindent
The good algorithmic properties of non-crossing $H$-graphs might hold beyond FO Model Checking, and \textsc{Feedback Vertex Set} seems a good candidate to consider. This is because \textsc{Feedback Vertex Set}, restricted to $H$-graphs, is in $\mathsf{XP}$~\cite{ChaplickTVZ21,Algo-H-graphs} but $\mathsf{W}[1]$-hard, when parameterized by $\Vert H\Vert$~\cite{JKT20}.

\begin{open}
Is {\sc Feedback Vertex Set} in $\mathsf{FPT}$ for proper $H$-graphs and non-crossing $H$-graphs when parameterized by $\Vert H\Vert$?
\end{open}

\noindent
Finally, recall that for every (simple) forest~$H$, we showed that 
the class of $H$-graphs is delineated (see Theorem~\ref{thm:effectiveDelineationHGraph}). What if $H$ is not a forest?

\begin{open}\label{o-6}
Is the class of $H$-graphs delineated if $H$ is a graph with a cycle? In particular, is the class of
$C_2$-graphs (i.e., circular arc-graphs) delineated?
\end{open}

\noindent
{\it Acknowledgments.} We thank Michał Pilipczuk for asking us a question at WG 2025 that led to Theorem~\ref{thm:effectiveDelineationHGraph}.

\end{document}